\documentclass[11pt]{article}
\pdfoutput=1
\usepackage[margin=1in]{geometry} 
\usepackage{amsmath,amsthm,amssymb,amsfonts,titling,bbm, titlesec, bm, physics, enumitem, accents, xcolor, setspace, fancyhdr, natbib, geometry,
	pdflscape}
\usepackage{graphicx}
\usepackage{subcaption}
\usepackage{float}
\usepackage[font=footnotesize,labelfont=bf]{caption}
\usepackage{array}
\usepackage[title]{appendix}
\usepackage{afterpage}
\usepackage{listings}
\usepackage{prodint}
\usepackage{algorithm}
\usepackage{algpseudocode}
\usepackage{algorithmicx}
\usepackage[hidelinks]{hyperref}

\newlength{\continueindent}
\usepackage{etoolbox}
\makeatletter
\newcommand*{\ALG@customparshape}{\parshape 2 \leftmargin \linewidth \dimexpr\ALG@tlm+\continueindent\relax \dimexpr\linewidth+\leftmargin-\ALG@tlm-\continueindent\relax}
\apptocmd{\ALG@beginblock}{\ALG@customparshape}{}{\errmessage{failed to patch}}
\makeatother

\algnewcommand\algorithmicstack{\textit{Stack:}}
\algnewcommand\Stack{\item[\algorithmicstack]}
\algnewcommand\algorithmicchoosetimegrid{\textit{Choose time grid:}}
\algnewcommand\Choosetimegrid{\item[\algorithmicchoosetimegrid]}
\algnewcommand\algorithmicchoosetimebasis{\textit{Choose time basis:}}
\algnewcommand\Choosetimebasis{\item[\algorithmicchoosetimebasis]}
\algnewcommand\algorithmicbuilddiscretedata{\textit{Build data at time $t_i$:}}
\algnewcommand\Builddiscretedata{\item[\algorithmicbuilddiscretedata]}
\algnewcommand\algorithmicfit{\textit{Fit:}}
\algnewcommand\Fit{\item[\algorithmicfit]}
\algnewcommand\algorithmicpredict{\textit{Predict:}}
\algnewcommand\Predict{\item[\algorithmicpredict]}

\newcommand\independent{\protect\mathpalette{\protect\independenT}{\perp}}
\def\independenT#1#2{\mathrel{\rlap{$#1#2$}\mkern2mu{#1#2}}}

\newcommand{\argmax}{\text{argmax}}
\newcommand{\argmin}{\text{argmin}}

\newtheorem{theorem}{Theorem}
\newtheorem{lemma}{Lemma}
\newtheorem{corollary}{Corollary}
\newtheorem{condition}{Condition}
\newtheorem{assumption}{Assumption}
\theoremstyle{definition}

\titleformat{\section}{\large\scshape\bfseries}{\thesection.}{1em}{}
\titleformat{\subsection}{\normalfont\bfseries}{\thesubsection.}{1em}{}

\usepackage{float, amsmath,amsfonts,graphicx,tikz}
\usepackage{times}
\usepackage{bm}
\usepackage{subcaption}
\usepackage{natbib}
\usepackage{enumitem}
\usepackage{algorithm}
\usepackage[affil-it]{authblk}
 \usepackage{xr} 
 \makeatletter
\newcommand*{\addFileDependency}[1]{
  \typeout{(#1)}
  \@addtofilelist{#1}
  \IfFileExists{#1}{}{\typeout{No file #1.}}
}
\makeatother

\newcommand*{\myexternaldocument}[1]{%
    \externaldocument{#1}%
    \addFileDependency{#1.tex}%
    \addFileDependency{#1.aux}%
}

\myexternaldocument{mainSupp}

\DeclareMathOperator*{\logit}{logit}

\title{Semiparametric logistic regression for inference on relative vaccine efficacy in case-only studies with informative missingness}

 \author[1, *]{Lars van der Laan}
 \author[2,3]{Peter B. Gilbert}

\affil[1]{%
    University of Washington, Statistics, 
Seattle, WA, USA.}
\affil[2]{%
   Fred Hutchinson Cancer Center, Biostatistics, 
Seattle, WA, USA.}
\affil[3]{%
  University of Washington, Biostatistics, Seattle, WA, USA.}
 \affil[*]{Corresponding author: lvdlaan@uw.edu}
\providecommand{\keywords}[1]
{
  \small	
  \textbf{\textit{Keywords---}} #1
}

\begin{document}
\maketitle
 
 \begin{abstract}

We develop semiparametric methods for estimating subgroup-specific relative vaccine efficacy against multiple viral strains in a partially vaccinated population. Focusing on observational case-only studies, we address informative missingness in strain type due to vaccination status, pre-vaccination characteristics, and post-infection factors such as viral load. We establish general conditions for the nonparametric identification of relative conditional vaccine efficacy between strains using covariate-adjusted conditional odds ratio parameters. Assuming a log-linear parametric form for strain-specific conditional vaccine efficacy, we propose targeted maximum likelihood estimators based on partially linear logistic regression, leveraging machine learning for flexible confounding adjustment. Finally, we apply our methods to estimate relative strain-specific conditional vaccine efficacy in the ENSEMBLE COVID-19 vaccine trial.

\end{abstract}

\section{Introduction}

For the development of vaccines against genetically diverse pathogens, ``sieve analysis" research seeks an understanding of how the protective efficacy of a vaccine depends on the genetic features of the disease-causing pathogen (Gilbert, Self, Ashby, 1998; Gilbert, Lele, Vardi, 1999; Follmann and Huang, 2018).\nocite{gilbert1998statistical,gilbert1999maximum,follmann2018sieve}  Understanding differential vaccine efficacy by pathogen genotype is important for regulatory approval and public health deployment of vaccines, as well as for guiding optimization of the pathogen strains to include in vaccine constructs. While sieve analysis is most rigorously based on randomized, placebo-controlled vaccine efficacy trials, once highly efficacious vaccines are available, sieve analysis must often be based on non-randomized observational studies, as is the case for SARS-CoV-2, the pathogen motivating this current methodological research (Rolland and Gilbert, 2021).\nocite{rolland2021sieve} 
Given the rapid emergence of the SARS-CoV-2 variants, such as delta and omicron,
maximally rigorous and robust statistical methods for sieve analysis of SARS-CoV-2 observational studies are needed.

Observational study designs that could be applied to SARS-CoV-2 vaccine sieve analysis include prospective and retrospective cohort designs, test-negative case-control designs, and various case-only designs (e.g., \cite{FiremanCasesOnly,dai2018case, patel2021evaluation}) where a ``case" is infection with SARS-CoV-2 diagnosed by RNA-PCR that may also require having certain
COVID-19 symptoms.
The current research focuses on case-only designs, which, after detecting cases via surveillance, measure the case-causing genotypes and ascertain the vaccine statuses of cases. 
The typical analysis method -- multivariable logistic regression -- compares vaccination status among disease cases with two different genotypes, where the regression aims to correct for confounding of the association of vaccination status with the case-causing genotype. This approach has been used for vaccines for several pathogens (e.g., Verani et al., 2015)\nocite{verani2015indirect} 
including for studies conducted through the Centers for Disease Control and Prevention's Emerging Infections Program. Three limitations of this analysis method are: (1) confounding control is achieved through parametric modeling such that model misspecification can bias estimation; (2) informative outcome/genotype missingness can bias parameter and variance estimation; and (3) the parametric estimator can be unstable or even ill-defined in settings with high-dimensional confounders.

 We robustify this approach by replacing 
standard logistic regression with a semiparametric model that can handle missing data and allows the conditional odds ratio to vary with participant covariates and time of infection detection. Specifically, we develop two targeted maximum likelihood estimators (TMLEs) \citep{vanderLaanRose2011} to estimate the relative strain-specific conditional vaccine efficacy, using machine learning to estimate nuisance parameters. The first estimator adjusts for outcome-missingness confounding by leveraging pre-vaccination baseline variables. In contrast, the second estimator adjusts for outcome-missingness confounding using both pre-and-post-infection variables, making it more versatile. As another novel contribution, we provide nonparametric identification results for an ideal conditional relative vaccine efficacy population estimand. We also demonstrate that under further causal assumptions, the conditional odds ratio parameter can identify a conditional relative efficacy parameter under a hypothetical vaccine intervention.

Our work builds upon previous studies on partially linear logistic regression models. Specifically, \cite{OddsRatioreadingsTMLE}, \cite{TchetgenOddsRatio}, and \cite{liu2021double} have all explored inference in such models without missingness. Of note, the first estimator we propose for pre-treatment informed missingness closely aligns with the TMLE algorithm for partially linear logistic regression models without missingness, as outlined in the unpublished review paper \cite{OddsRatioreadingsTMLE}. Another study relevant to our research is \cite{tchetgen2010semiparametric}, which introduces a semiparametric estimator for examining the statistical interaction between genetic and environmental factors on the risk of a dichotomous disease status in a case-only design. However, our focus differs from theirs, as we are primarily interested in the conditional odds ratio between treatment and outcome rather than statistical interactions between two covariates on the outcome. Additionally, our study accounts for outcome missingness and utilises machine learning for nuisance estimation, which are important considerations for our research.

Our methods are particularly useful in two types of scenarios: (1) data-rich settings that can benefit from the use of flexible and powerful machine-learning tools, and (2) high-dimensional settings where ordinary logistic regression may yield unreliable results. For certain machine-learning algorithms, such as random forests \citep{breiman2001random} and gradient-boosting \citep{friedman2001greedy}, thousands of samples with hundreds of endpoints may be necessary for achieving satisfactory performance. However, other algorithms such as regularized logistic regression \citep{tibshirani1996regression}, multivariate adaptive regression splines \citep{friedman1991multivariate}, generalized additive models \citep{hastie1987generalized}, and the highly adaptive lasso \citep{HAL2016} can achieve good results with only hundreds of samples if properly tuned. Our method allows for variable-selection techniques and lasso-regularised logistic regression \citep{tibshirani1996regression} to estimate nuisance components of the data-generating distribution. As a result, our approach allows for inference even when data-driven variable selection is performed for confounding adjustment.
 


\section{Problem formulation and population estimand of interest}
\subsection{Setup and data-structure}

We consider a target population $\mathcal{P}$ of $N\approx \infty$ independently and identically distributed individuals, such as all inhabitants of a city or all inhabitants of a city not previously infected by the pathogen. Suppose that individuals in $\mathcal{P}$ are at risk of being infected by a viral disease with several different viral variants, and a vaccine is available to all individuals in $\mathcal{P}$ on an opt-in basis. To assess the vaccine's relative performance against different viral strains, an observational monitoring site offering optional viral-infection testing services is open to members of $\mathcal{P}$ starting from time $t=0$, which we define as the monitoring site's opening date. When an individual tests positive for infection at the monitoring site, i.e. becomes a case, various data are recorded, including individual features such as vaccination status, and features of the virus infection, such as viral load and virus strain type. No data are recorded from individuals who test negative for infection. We assume that the time between the initial infection and diagnosis at the monitoring site is negligible, so the time of infection is known with certainty if an infection is diagnosed. Targeted to our ENSEMBLE application context, we also assume that no one in $\mathcal{P}$ is infected before time $t=0$ (otherwise, we would redefine the population).

We define $\mathbb{T}$ as the set of all possible infection times. Each individual in the target population $\mathcal{P}$ is represented by the unobserved (full) data structure $O_F \equiv (W, A, T, W_T, J) \sim P_F$, where $W \in \mathbb{R}^d$ denotes baseline covariates determined before the base time $t=0$, $A$ is a binary indicator taking value 1 if the individual was vaccinated before the time of infection, $T$ is the time of the individual's first viral infection after the monitoring site at $t=0$, $W_T$ represents post-infection covariates measured upon visiting the monitoring site, and $J$ is a binary feature of the viral infection that only applies to infected individuals (a ``mark" feature). We also define the time-dependent vaccination status indicator $A(t)$, which takes the value 1 if the individual is considered vaccinated immediately before time $t$. We allow $A(t)$ to change from 1 to 0, modeling vaccine efficacy waning coarsely. In this manuscript, we assume $J$ represents the indicator variable that the infection-causing virus is of a specific strain, and assume that multiple distinct strains do not cause infection, which is supported for COVID-19. Examples of baseline covariates $W$ include age, sex, history of previous infection, and living location. If $A$ is a binary variable representing two different vaccines, $W$ may also include time since vaccination and other shared vaccine characteristics. 

\begin{figure}[H]
    \centering
    \includegraphics[width=0.7\linewidth]{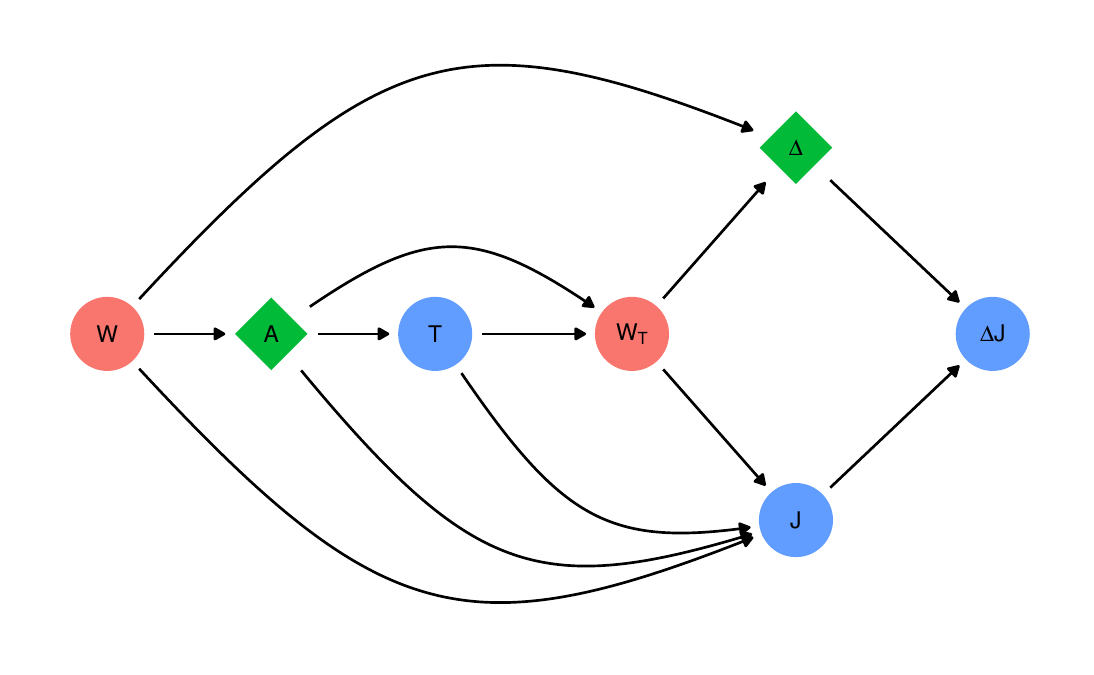}
    \caption{A directed acyclic graph (DAG) representing a possible causal model for the underlying data-generating process with strain missingness informed by post-treatment variables. Confounders are colour-coded red, interventional variables are colour-coded green and diamond-shaped, and outcomes are colour-coded blue.}
    \label{fig::DAG}
\end{figure}

As individuals have the option to visit the monitoring site for testing and we only observe some cases, we cannot observe the full data structure $O_F$ for all infected members of the target population. Furthermore, baseline variables, vaccine status, and time-of-infection may inform this missingness. Even for individuals whose infection is observed, the outcome variable $J$ may not be observed due to measurement issues or resource limitations. The missingness of $J$ may also be influenced by post-infection covariates $W_T$, which can be causally affected by vaccine assignment or infection time $T$. For example, failed viral sequencing due to low viral load at infection time could lead to missing $J$. We use the indicator variable $R$ to denote whether an individual's viral infection is observed at the monitoring site and $ \Delta$ to denote whether the post-infection outcome variable $J$ is recorded. Therefore, the observed case-only data structure is $O \equiv R(W,A,T,W_T,\Delta,\Delta J) \sim P$, where $O$ equals $(W,A,T,W_T,\Delta,\Delta J)$ if $R=1$ and is an empty set otherwise. Also, $\Delta J = J$ if $\Delta = 1$, and $\Delta J = 0$ otherwise. We present one possible causal model for the data-generating process in Figure \ref{fig::DAG}.


The methods presented in this manuscript are not dependent on the specific interpretation of $R$, $W$, $A$, $T$, $W_T$, $\Delta$, and $J$, and can be applied to many other scenarios. For example, in addition to baseline variables $W$, post-infection variables such as symptoms and severity of the infection may be included in $W$. However, since vaccination status may influence these variables, caution should be exercised in their causal interpretation. Other applications may focus exclusively on vaccinated individuals and define $A$ in terms of the type of vaccine, the number of vaccinations, or the antibody response to vaccination.



 \subsection{Notation}

 Recall that $P_F$ and $P$ are the data-generating distributions of $O_F$ and $O$, respectively. Let $\|\,\cdot\,\|$ denote the $L^2(P)$ norm. For a function $f \in L^2(P)$, we use the empirical process notation $Pf = \int f dP$ and $P_n f = \int f dP_n$ where $P_n$ denotes the empirical measure induced by a sample of $n$ observations of $O$. Throughout this text, we denote a possible realisation of the random variable $O= R(W, A, T, W_T, \Delta, \Delta J)$ by the lowercase letters $o = (a, w, t, w_T, \delta, \delta j)$. For a given distribution $P'$ and realisation $o$ of $O$, denote the conditional probability of the reference strain given its parents as $\overline{\mu}_{P'}(a,w,t,w_T) := {P'}(J=1| \Delta = 1,R=1, W_T = w_T,A=a, W=w, T=t)$. Denote the conditional vaccination probabilities as $\pi_{P'}(w,t) := {P'}(A=1|W=w, R=1, T=t)$ and $\widetilde{\pi}_{P'}(w,t) := {P'}(A=1\mid R=1, \Delta=1, W=w,T=t)$, and denote the conditional strain nonmissingness probability as $\Pi_{P'}(a,w,t,w_T) := {P'}(\Delta = 1 \mid R=1, W_T = w_T, A=a, W=w, T=t)$. Denote ${\mu}_{P'}(a,w, t) := {P'}(J=1|R=1, A=a, W=w, \Delta = 1, T=t)$ and ${\mu}_{adj,P'}(a,w,t) := E_{P'}[ \overline{\mu}_{P'}(a,w, t,w_T)|R=1, A=a,W=w, T=t]$. We define the conditional odds ratio $OR(P')(w,t)$ and adjusted conditional odds ratio $OR_{adj}(P')(w,t)$ by
\begin{align*}
OR(P')(w,t) &:= \frac{\mu_{P'}(1,w,t)/(1-\mu_{P'}(1,w,t))}{\mu_{P'}(0,w,t)/(1-\mu_{P'}(0,w,t))} \\
OR_{adj}(P')(w,t) &:= \frac{\mu_{adj, P'}(1,w,t)/(1-\mu_{adj, P'}(1,w,t))}{\mu_{adj, P'}(0,w,t)/(1-\mu_{adj, P'}(0,w,t))} .
\end{align*}

\subsection{The relative strain-specific conditional vaccine efficacy and semiparametric model assumptions}
For strain type $j \in \{0,1\}$ and realisation $(w,a)$ of $(W,A(t))$ at time $t\in\mathbb{T}$, let 
$$\lambda_j(t \mid A(t)=a, W=w) := \lim_{h \downarrow 0} \frac{1}{h}P_F(J=j, T \in [t, t+h)\mid A(t) = a,\, W=w,\, T \geq t)$$

\noindent be the cause-specific conditional hazard function. This function represents the probability of an individual acquiring the $j$th viral strain at time $t$, given that they had vaccination assignment $a$ before time $t$ and were not infected before time $t$, based on their baseline information. We also define the strain-specific conditional vaccine efficacy for each strain type $j \in \{0,1\}$ as $VE_j(t \mid w) = 1- \exp \left\{ \theta_j(t\mid w)\right\}$, where $\theta_j(t\mid w) = \log \left\{ \lambda_j(t \mid A(t)=1, W=w) \right. /$ \newline $ \left. \lambda_j(t \mid A(t)=0, W=w)\right\}$.

Suppose we observe $n$ i.i.d. case-only instances $(O_1, O_2, \dots, O_n)$ of the data-structure $O \mid (R=1) \sim P  $, and we wish to compare the efficacy of a vaccine against strain $J=1$ with that against strain $J=0$. Since we can only observe individuals who have contracted the viral infection, i.e. cases, we cannot determine the strain-specific conditional vaccine efficacy parameters $VE_0(t \mid w)$ and $VE_1(t \mid w)$ from the data-generating distribution $P$ of the observed data. As such, we instead consider the relative strain-specific conditional vaccine efficacy parameter,
\begin{equation}
\label{eqn::relVacEff}
\Psi_F(P_F)(w,t) := \frac{1-VE_1(t \mid w)}{1-VE_0(t \mid w)}.
\end{equation}
We can express this parameter as the conditional odds ratio between $J$ and $A$ under the target population: $\Psi_F(P_F)(w,t) = OR_F(P_F)(w,t)$ where
\begin{align}
& \hspace{-0.5cm} OR_F(P_F)(w,t) := \frac{P_F(J=1|A=1, W=w, T = t)/P_F(J=0|A=1, W=w, T=t)  }{P_F(J=1|A=0, W=w, T=t) /P_F(J=0|A=0, W=w, T=t)},  \label{eqn::relVacEffOddsRatio}
\end{align}
which follows from Bayes rule and that $A(t) = A$ on the event $\{T=t\}$. This identity implies that the estimand $\Psi_F(P_F)(w,t)$ is determined by the conditional distribution of the strain type among population cases, given baseline covariates $W$, infection time $T$, and vaccination status at the time of infection. Under certain missing-data assumptions, we can identify $\Psi_F(P_F)$ from the data-generating distribution $P$ of the observed data.

In this manuscript, we develop semiparametric and asymptotically linear estimators for $OR_F(P_F)$ under two sets of missing-data assumptions for $P_F$ given in the next section, recognising that this provides estimators for the relative strain-specific conditional vaccine efficacy, $(w,t) \mapsto \Psi_F(P_F)(w,t)$. We assume that $OR_F(P_F)$, and thus $\Psi_F(P_F)$, is log-linear with a functional form that is known up to a finite-dimensional coefficient vector $\beta_{F}(P_F)$. That is, we assume that there exists a {known} vector-valued function $\underline{f}: \mathbb{R}^d\times \mathbb{T} \mapsto \mathbb{R}^s$ and an unknown coefficient vector $\beta_{F}(P_F) \in \mathbb{R}^s$ such that
\begin{equation}
    \label{eqn::parametricForm}
    \log OR_F(P_F)(w,t) = \beta_F(P_F)^T \underline{f}(w,t).
\end{equation}
Equation (\ref{eqn::parametricForm}) implies a semiparametric model assumption on $P_F$, which is the same as a partially linear logistic regression model on the outcome regression function $P_F(J=1 \mid A = a, W = w, T = t)$, as described by Tchetgen Tchetgen et al. (2010). Thus, we can replace the model assumption in Equation (\ref{eqn::parametricForm}) with the following equation:
\begin{equation}
    \text{logit}\left\{P_F(J=1 \mid A = a, W = w, T = t) \right\} = a\beta_F(P_F)^T \underline{f}(w,t) + h_{P_F}(w,t),\label{eqn::partiallyLinear}
\end{equation}
where $h_{P_F}(w,t) := \text{logit}\left\{P_F(J=1 \mid A = 0, W = w, T = t) \right\}$ remains unspecified.

Nonparametric modeling of $(w,t) \mapsto OR_F(P_F)(w,t)$, while possibly of interest for some applications, makes interpretation and inference difficult. This is because the parameter $P_F \mapsto OR_F(P_F)(w,t)$ is not pathwise differentiable under a nonparametric statistical model, which means that root-$n$ consistent, regular, and asymptotically-normal estimators cannot be obtained without additional assumptions on the functional form of $OR_F(P_F)$ (Bickel et al., 1993). On the other hand, parametric modeling of $OR_F(P_F)$ allows for the development of root-$n$ consistent and semiparametric efficient estimators using techniques from efficiency theory (Bickel et al., 1993; van der Laan, Robins, 2003; van der Laan, Rose, 2011), while still leaving nuisance components of the data-generating distribution unspecified. Fortunately, the relative conditional vaccine efficacy is an interpretable feature whose functional form may be plausibly modeled using domain knowledge. For instance, it might be known based on the analysis of vaccine clinical trials that the strain-specific vaccine efficacy is close to constant across different baseline covariate subgroups, i.e $\log OR_F(P_F)(w,t) \approx \beta_{F}(P_F)$ for some $\beta_F(P_F) \in \mathbb{R}$. Regardless, $\log OR_F(P_F)(w,t)$ could still be relatively flexibly modeled using, for instance, a finite number of cubic spline basis functions. We make no further assumptions on the functional form of $P_F$ to alleviate unnecessary biases, since usually little is known about the functional form of other features of $P_F$ and $P$.

Before proceeding, it is important to note that we focus on the inference of a conditional version of the relative strain-specific vaccine efficacy. This is intentional, as marginal versions of vaccine efficacy are often not identifiable due to potential differences in the covariate distribution between observed cases and the overall population. Moreover, utilizing the conditional relative strain-specific vaccine efficacy enables the discovery of vaccine subgroup effects, and it may be more transportable to populations with varying baseline covariate distributions.

\section{Nonparametric identification of the target parameter}
In view of the missingness indicators $R$ and $\Delta$, the observed data-structure $O$ is a coarsening of the full data-structure $O_F$. As such, the target estimand $OR_F(P_F)$ is not necessarily identified from the probability distribution $P$ of $O$. In this section, we introduce causal (missing-data) assumptions that enable identification of $OR_F(P_F)$ from the data-generating distribution of $O$. Moreover, we provide a set of (strong) causal assumptions on the vaccine efficacy mechanism that allows for interpretation of $OR_F(P_F)$ as a measure of relative vaccine efficacy under a hypothetical vaccine intervention on the target population. These identification results are nonparametric and do not require assumptions about the functional forms of $P$ and $P_F$.

\subsection{Identification of population parameter from observed data distribution}

In this section, we establish conditions for identifying the conditional odds ratio in (\ref{eqn::relVacEffOddsRatio}), and thus the relative conditional vaccine efficacy parameter in (\ref{eqn::relVacEff}), from a parameter in the observed data-generating distribution. We provide two sets of causal assumptions that enable identification, with the second set being a weaker version of the first. The first set assumes that strain missingness is informed only by baseline variables and vaccination status, while the second set allows post-infection variables also to influence strain missingness. We present the first identification result below.

\begin{assumption}
\textit{Weak overlap for vaccination status:} $P\left(1 > P(A= 1| W,R=1) > 0\right)=1$.\label{cond::ident::C0}
\end{assumption}
\begin{assumption}
\textit{Weak overlap for strain missingness:} $P\left(P(\Delta = 1|A, W,R=1, T=t) > 0\right)=1$.\label{cond::ident::C1}
\end{assumption}
\begin{assumption}
\textit{Weak overlap for case missingness:} $P_F\left(P_F(R = 1| A, W, T=t) > 0\right)=1$.\label{cond::ident::C2}
\end{assumption}
\begin{assumption}
\textit{Cases are missing at random in the population:} $R \independent J | A, W, T =t$. 
\label{cond::ident::C4}
\end{assumption}
\begin{assumption}
\textit{Strain is missing at random among observed cases:} $\Delta \independent J | A, W,T = t, R=1$. \label{cond::ident::C3}
\end{assumption}

\begin{theorem}
Under assumptions \ref{cond::ident::C0}-\ref{cond::ident::C3}, the target estimand $OR_F(P_F)(w,t)$ defined at (\ref{eqn::relVacEffOddsRatio}) is identified from $P$ by the conditional odds ratio,
$$OR(P)(w,t) := \frac{\mu_P( a=1, w, t)/(1-\mu_P(a=1, w, t))}{\mu_P(a=0, w, t) / (1-\mu_P(a=0, w, t))}.$$
\label{theorem::ident}
\end{theorem}


 We now discuss the assumptions of Theorem \ref{theorem::ident}. Assumptions \ref{cond::ident::C0}-\ref{cond::ident::C2} are standard positivity assumptions that ensure that $OR(P)$ and $OR_F(P_F)$ are well-defined. The first assumption guarantees a positive probability of observing each vaccination status among the cases in each stratum of $W$. The second assumption stipulates that among the observed cases ($R=1$) infected at time $T=t$, there is a positive probability of observing $J$ in all strata of $(A,W)$. Similarly, the third assumption requires that among all members of $\mathcal{P}$ who were infected at time $T=t$ there was a positive probability of visiting the monitoring site, being tested positively for infection, and having the data-structure $(W,A,T)$ recorded. Assumption \ref{cond::ident::C3} states that the outcome $J$ is missing-at-random among observed infections at time $T=t$ within strata of $A$ and $W$. This assumption is violated if the missingness of $J$ is caused by post-infection variables that predict $J$. In such cases, adjusting for these variables by including them in $W$ may not lead to a causally interpretable relative conditional vaccine efficacy measure. Assumption \ref{cond::ident::C4} is a strong assumption and is essential for identifying the relative conditional vaccine efficacy of the population. It requires that, at time $t$ and within all strata of $A$ and $W$, the distribution of the viral strain type is identical among visitors and non-visitors to the monitoring site. 
 This assumption would be violated if different viral strains cause different symptoms, affecting an individual's decision to visit the monitoring site for testing. In this case, the strain type $J$ causally affects its missingness $R$, inducing confounding bias that may not be adjustable using the observed data. To increase the plausibility of Assumptions \ref{cond::ident::C2} and \ref{cond::ident::C4}, we can redefine the population $\mathcal{P}$ as individuals who are likely to visit the monitoring site for testing or redefine the endpoint as severe viral infection. For example, suppose individuals with severe symptoms are less likely to consider their symptoms when deciding whether or not to visit the monitoring site. In that case, a possible solution is to redefine the endpoint $J$ as an indicator of being infected with a particular viral strain and having a symptom severity level exceeding a predetermined threshold.

\vspace{0.5cm}

We present assumptions that relax Assumption \ref{cond::ident::C3} and enable identification of $\Psi_F(P_F)$ in a broader context where the post-infection covariates $W_T$, which also inform $J$, determine the missingness of $J$.

\begin{assumption}
 \textit{Weak overlap for strain missingness:} $P_F\left(P(\Delta = 1|W_T, A, W,R=1, T=t) > 0\right)=1$. \label{cond::ident::D1}
 \end{assumption}
\begin{assumption}
\textit{Strain is missing at random after adjusting for $W_T$:} $\Delta \independent J | W_T,A, W,   R=1,T =t$. \label{cond::ident::D3}
\end{assumption}



\begin{theorem}
Suppose assumptions \ref{cond::ident::C0}, \ref{cond::ident::C2}, \ref{cond::ident::C4} and \ref{cond::ident::D1} and \ref{cond::ident::D3} hold. Then, the target estimand (\ref{eqn::relVacEff}) is identified from $P$ by the adjusted conditional odds ratio,
$$OR_{adj}(P)(w,t) := \frac{\mu_{adj,P}(a=1, w, t)/(1-\mu_{adj,P}(a=1, w, t))}{\mu_{adj,P}(a=0, w, t) / (1-\mu_{adj,P}(a=0, w, t))}.$$
 
\label{theorem::ident2}
\end{theorem}

Assumption \ref{cond::ident::D1} requires that assumption \ref{cond::ident::C1} remains true after conditioning on the post-infection variables $W_T$. Assumption \ref{cond::ident::D3} is a weaker version of Assumption \ref{cond::ident::C3}, as it allows for the possibility that the missingness of $J$ may be informed by post-infection variables $W_T$ even if these variables causally affect $J$. 
Notably, the DAG depicted in Figure \ref{fig::DAG} implies that the corresponding causal model satisfies Assumption \ref{cond::ident::D3} due to the conditional independence relations it induces. In specific scenarios, such as when viral strain missingness is related to factors like low viral load or the severity of infection symptoms, it may be plausible to satisfy Assumption \ref{cond::ident::D3} by incorporating post-infection variables that capture these factors into the post-infection covariate $W_T$.

 \subsection{Identification of population parameter with causal parameter}

To make the relative conditional vaccine efficacy estimand in Equation (\ref{eqn::relVacEff}) more causally interpretable, we would like to adjust for baseline variables that predict both vaccination status at time $t$, $A(t)$, and the outcomes $(T,J)$. This adjustment would make the estimand $\Psi_F(P_F)$ more transportable across populations with different distributions of baseline covariates and vaccination status. However, because individuals in the real-world population are vaccinated at different times, and their vaccination time may influence both the infection time and the viral strain that causes the infection, $\Psi_F(P_F)$ is generally not interpretable as relative conditional vaccine efficacy from a hypothetical randomized control vaccine trial. In this section, we introduce potentially strong causal assumptions that can be used to show that $\Psi_F(P_F)(w,t)$ corresponds with an instantaneous conditional relative vaccine efficacy under a hypothetical vaccine intervention.


To this end, we make the simplifying assumption that the possible infection times are discrete and integer-valued, e.g., $\mathbb{T} \subset \mathbb{N}$ encodes days passed since baseline $t=0$. Consider a member $O_F \in \mathcal{P}$ of the population that has not acquired viral infection before time $t$ and is not yet vaccinated at time $t$, i.e., $A(t-1) := 0$. Let $(T_t(1),J_t(1))$ be the potential outcomes \citep{Rubin2005} of $(T,J)$ that would be observed under the hypothetical intervention that vaccinates the individual $O_F$ at time $t$. Moreover, let $(T_t(0),J_t(0))$ be the potential outcomes that would be observed if the individual is not vaccinated at time $t$. We assume the causal ordering that $A(t)$ is determined before the potential outcomes $(T_t(a), J_T(a): a \in \{0,1\})$. We stress that the potential outcomes $T_t(a)$ and $J_t(a)$ are only meaningfully defined on the event $\{T \geq t, A(t-1) = 0\}$. Let $O_{F,c} = (O_F, A(t-1), T_t(0), J_t(0), T_t(1), J_t(1)) \sim P_{F,c}$ be the corresponding causal data-structure. Under the following causal assumptions, Theorem \ref{theorem::CausalIdent} establishes that $OR_F(P_F)(w,t) = \Psi_F(P_F)$ can be interpreted as relative conditional vaccine efficacy under a hypothetical randomized controlled vaccine efficacy trial performed among all members of the population who have not yet been infected or vaccinated before time $t$.
 
\begin{assumption}
\textit{Consistency of potential outcomes:} For each $a \in \{0,1\}$, \newline 
$(T_t(a), J_t(a)) \mid \{A(t)=a,\, A(t-1) = 0,  \,T \geq t,\, W\} \,{\buildrel d \over =}\, (T,J)  \mid \{A(t)=a,\, A(t-1) = 0, \, T \geq t,\, W\}.  $
\label{cond::Causalconsistency}
\end{assumption}
\begin{assumption}
\textit{Exchangeability:} For each $a \in \{0,1\}$, $(T_t(a), J_t(a)) \independent  A(t) \mid \{  A(t-1) = 0,\, W,\, T \geq t\}.$\label{cond::CausalExchange}
\end{assumption}
\begin{assumption}
\textit{Markov dependence of outcomes on vaccination history:}  For each $a \in \{0,1\}$, $  (T,J) \independent  A(t-1)  \mid \{A(t)=a,\, T \geq t,\, W\} .$
\label{cond::indepVacTime}
\end{assumption}

 \begin{theorem}
 Under Assumptions \ref{cond::Causalconsistency}-\ref{cond::indepVacTime}, we have $OR_F(w,t) =$
  $$\frac{P_{F,c}(T_t(1) = t, J_t(1) = 1 \mid T \geq t, A(t-1)=0, W =w)/P_{F,c}(T_t(0) = t, J_t(0) = 1 \mid T \geq t, A(t-1)=0, W =w)}{P_{F,c}(T_t(1) = t, J_t(1) = 0 \mid T \geq t, A(t-1)=0, W =w)/P_{F,c}(T_t(0) = t, J_t(0) = 0 \mid T \geq t, A(t-1)=0, W =w)}$$ \label{theorem::CausalIdent}
 \end{theorem}
  

Assumption \ref{cond::Causalconsistency} is a standard consistency condition for the potential outcomes \citep{Rubin2005, Pearl2009}. Assumption \ref{cond::CausalExchange} assumes no unmeasured confounders between the potential outcomes and vaccination assignment at time $t$, which is also standard in the causal inference literature \citep{Rubin2005, Pearl2009}. Assumption \ref{cond::indepVacTime} is a strong assumption that assumes the distribution of $(T,J)$ among uninfected individuals who receive the vaccine at time $t$ is the same as the distribution among those who had already been vaccinated prior to time $t$. This implies that the likelihood of a vaccinated person acquiring an infection at time $t$ is not affected by the duration since vaccination and that the vaccine's complete protective effect is realised at the time of administration. While these assumptions are typically not met, the Markov-type assumption may be more plausible if previously vaccinated individuals are considered unvaccinated after sufficient time has passed. In addition, if vaccination is replaced with a monoclonal antibody, it is more plausible that the effect of the treatment is realised quickly after its administration. We caution against drawing conclusions from this causal identification beyond that it motivates the benefit of adjusting for covariates predictive of both vaccination and the outcomes when interpreting $\Psi_F(P_F)$.

\section{Semiparametric targeted learning and inference for the conditional odds ratios}

The following sections give efficient and inefficient influence functions for the odds ratio coefficient parameters $\beta$ and $\beta_{adj}$. Using the targeted maximum likelihood estimation (TMLE) framework (van der Laan, Rose, 2011)\nocite{vanderLaanRose2011}, we develop $\sqrt{n}$-consistent and asymptotically normal substitution estimators for both the coefficient vectors and their respective conditional odds ratios. An R package implementing our methods can be found on GitHub at: \hyperlink{https://github.com/Larsvanderlaan/spCaseOnlyVE}{https://github.com/Larsvanderlaan/spCaseOnlyVE}.

The efficient influence function plays a key role in semiparametric inference by determining the best possible variance of asymptotically linear and regular estimators. This function encodes the sensitivity of the target estimand under perturbations of the data-generating distribution and can be used to construct asymptotically linear and locally efficient estimators. Popular semiparametric estimation strategies include one-step estimation (Bickel et al., 1993)\nocite{bickel1993efficient} and influence function-based estimating equations (Robins et al., 1994, van der Laan and Robins, 2003; Chernozhukov et al., 2018)\nocite{vanderlaanunified}\nocite{robinsCausal}\nocite{DoubleML}\nocite{cvTMLEDoubleML}, as well as targeted maximum likelihood estimation (van der Laan and Rubin, 2006\nocite{laan_rubin_2006}; van der Laan, Rose, 2011)\nocite{vanderLaanRose2011}. However, one-step and estimating equation approaches are not substitution estimators and, therefore, may not respect known constraints of the statistical model, which can impact finite sample performance (Kang and Schafer, 2007)\nocite{KangSchaferAIPW}. TMLE provides a general framework for constructing efficient substitution estimators that agree with the one-step and estimating-equation-based estimators obtained from the targeted data-generating distribution estimator while potentially improving finite-sample performance (Porter et al., 2011\nocite{PorterLaanPerform}). The targeted learning methodology \citep{vanderLaanRose2011,cvTMLEDoubleML} advocates for using flexible, data-adaptive, black-box machine-learning algorithms within the estimation procedure, avoiding bias due to misspecification and relaxing the conditions needed for asymptotic linearity and efficiency.

\subsection{Semiparametric estimation when pre-vaccination variables inform strain missingness}

 \label{section::spTMLE}

 In this section, we provide asymptotically linear and semiparametrically efficient TMLEs for the population odds ratio $OR(P)$ and the parameter vector $\beta(P)$, assuming the semiparametric assumption that $OR(P)(w,t) = \exp {\beta(P)^T \underline{f}(w,t)}$ for some $\beta(P) \in \mathbb{R}^s$. Our semiparametric assumption is equivalent to assuming that the relative conditional vaccine efficacy parameter $\Psi_F(P_F)$ satisfies \eqref{eqn::parametricForm} with $\beta_F(P_F) = \beta(P)$, under the identification result of Theorem \ref{theorem::ident}, which assumes that the case missingness $R$ and strain missingness $\Delta$ are randomized conditional on baseline covariates and treatment. In the next section, we propose a TMLE that can provide valid inference even when post-infection variables inform the strain missingness.

 To this end, let $\mathcal{M}$ be a semiparametric statistical model corresponding with the assumption that each $P' \in \mathcal{M}$ satisfies $ \mu_{P'}(a,w,t) =  \text{expit}\{ a\cdot \beta(P')^T\underline{f}(w,t) + h_{P'}(w,t)\}$ for some $\beta(P') \in \mathbb{R}^s$ and $h_{P'}: \mathbb{R}^{d} \times \mathbb{T} \rightarrow \mathbb{R}$. Under the assumption that $P \in \mathcal{M}$, we have $OR(P)(w,t) = \beta(P)^T\underline{f}(w,t)$. We first present the efficient influence function of the parameter $P \mapsto \beta(P)$ under the semiparametric model assumption that $P \in \mathcal{M}$. Denote, for each realisation $(a,w,t)$ of $(T,A,W)$ and $P' \in \mathcal{M}$, the conditional variance function $\sigma^2_{P'}(a,w,t) := \mu_{P'}(a,t,w) (1 -\mu_{P'}(a,w,t) )$. Further, define 
$$ H_{P'}(a,w, t) :=   \left(a  -\frac{ E_{P'}[ A \sigma^2_{ P'}(a=1,w,t)\mid  T = t, W = w, \Delta = 1]
}{ E_{P'}[  \sigma^2_{ P'}(A,w,t)  \mid  T = t, W = w, \Delta = 1]}\right),$$
which can be seen to implicitly depend on $\widetilde{\pi}$ after expanding the expectation over $A$. Next, define the scaling matrix $\Lambda_{P'}$ as 
$$\Lambda_{P'}^{-1} := E_{P'}\left\{ \underline{f}(W,T)\underline{f}^T(W,T) \frac{\Delta (1-\widetilde{\pi}_{P'}(W,T))\widetilde{\pi}_{P'}(W,T) \sigma^2_{ P'}(1,W,T)\sigma^2_{ P'}(0,W,T)}{ (1-\widetilde{\pi}_{P'}(W,T))\sigma^2_{ P'}(0,W,T) +\widetilde{\pi}_{P'}(W,T)\sigma^2_{ P'}(1,W,T) }\right\},$$
where we assume throughout that $\Lambda_{P'}^{-1}$ is invertible.

The efficient influence function of $\beta$ under $\mathcal{M}$ is given in the following theorem and its derivation follows with minor modifications from the proof of the semiparametric efficient influence function for the partially linear logistic regression model given in van der Laan (2009)\nocite{OddsRatioreadingsTMLE} -- see also Tchetgen Tchetgen et al. (2010).

\begin{theorem}
 
The vector-valued efficient influence function of $\beta$ at $P \in \mathcal{M}$ with respect to the semiparametric statistical model $\mathcal{M}$ is
$$ D_{P}(o) :=   (\Lambda_{P} \underline{f}(w,t)) \odot \delta H_P(a,w,t) \left[\delta j  - \mu_{P}(a,w,t)\right], $$
where the product ``$\odot $" is taken coordinate-wise.
\label{theorem::EIF::uninformative}
\end{theorem}

Next, let $P_{n,0} \in \mathcal{M}$ be an estimate of $P$, which can be obtained using the methods described in Appendix \ref{section::pluginestimator} of the Supplementary Information. Using techniques from efficiency theory \citep{bickel1993efficient}, we can show that the bias of the plug-in estimator $\beta(P_{n,0} )$ is equal to $-P D_{P_{n,0}}$, up to typically second-order terms. This suggests using the influence function $D_{ P_{n,0}}$ to construct a debiased estimator of $\beta(P)$. One such estimator for $\beta(P)$ is the one-step estimator (Bickel et al., 1993) which is given by $\beta(P_{n,0})  + P_n D_{P_{n,0}}$, where the second term is a bias correction obtained by empirically estimating $P D_{P_{n,0}}$. 


 An alternative estimator, TMLE, constructs an updated estimator $P_{n}^* \in \mathcal{M}$ from $P_{n,0}$ corresponding with updated nuisance such that the efficient score equation is solved: 
  \begin{equation}
  \label{eqn::tmle1score}
      \frac{1}{n}\sum_{i=1}^n D_{P_{n}^*}(O_i) = 0.
    \end{equation}
  Since the efficient score equation corresponds with the debiasing term for the one-step estimator of $\beta(P)$ obtained using $P_{n}^*$, the resulting TMLE $\beta(P_{n}^*)$ is both a one-step estimator and a substitution estimator. Leveraging this fact, we show later, under mild regularity conditions, that the TMLEs $OR(P_{n}^*)(w,t) = \exp\{\underline{f}(w,t)^T \beta(P_{n}^*)\}$ and $\beta(P_{n}^*)$ for $OR(P)$ and $\beta(P)$ are asymptotically linear and efficient estimators. 

The TMLE algorithm involves performing iterative (working) maximum likelihood estimation along a (data-dependent) parametric submodel that fluctuates the initial estimator $P_{n,0} \in \mathcal{M}$ in a direction determined by the efficient influence function of $\beta$, a so-called least favorable model \citep{vanderLaanRose2011}. In this case, the submodel only fluctuates the conditional train probability $\mu_{P_{n,0}}$ and leaves all other components of $P_{n,0}$ unchanged. 

For a given offset $P' \in \mathcal{M}$, define the (least-favorable) fluctuation submodel through $P'$ and working log-likelihood function:
\begin{itemize}
    \item \textit{Logistic fluctuation submodel:} $(P'(\varepsilon): \varepsilon \in \mathbb{R}^s) \subset \mathcal{M}$ is a submodel that only fluctuates $\mu(P')$ and satisfies
    $$\varepsilon \mapsto \text{logit}\left\{\mu_{P'(\varepsilon)}(a,w,t) \right\} = \text{logit}\left\{\mu_{P'}(a,w,t) \right\} + \varepsilon^T \underline{f}(w,t) H_{P'}(a,w,t): \varepsilon \in \mathbb{R}^s$$
    \item \textit{Working log-likelihood function:} 
    $$L_n(P') :=  \frac{1}{n}\sum_{i=1}^n   \Delta_i \left\{J_i \log \mu_{P'}(Z_i) + (1-J_i) \log(1 - \mu_{P'}(Z_i)) \right\} \text{ where } Z_i := (  A_i, W_i, T_i).$$
\end{itemize}

By construction, the submodel $ \{P'(\varepsilon) : \varepsilon \in \mathbb{R}^{s} \}$ respects the partially linear constraints imposed by $\mathcal{M}$. In particular, we have $ \mu_{P'(\varepsilon)}(a,w,t) = \text{expit}\left\{ h_{P'(\varepsilon)}(w,t) + a \beta(P'(\varepsilon))^T \underline{f}(w,t)\right\}$ where $\beta(P'(\varepsilon)) = \beta(P') + \varepsilon $ and $ h_{P'(\varepsilon)} = h_{P'}(w,t) + \varepsilon^T \underline{f}(w,t) H_{P'}(a=0,w,t)$. It can also be verified that the score vector of the working log-likelihood along this submodel satisfies
 \begin{align*}
       \frac{d}{d\varepsilon} L_{n}(P'({\varepsilon}))   & = \frac{1}{n}\sum_{i=1}^n \Delta_i H_{P'}( A_i, W_i, T_i) \left[J_i -  \mu_{P'({\varepsilon})}(A_i,W_i,T_i)\right]. 
\end{align*}
To construct a targeted MLE $P_n^*$ satisfying Equation \eqref{eqn::tmle1score}, we note that $\frac{d}{d\varepsilon} L_n(P'(\varepsilon)) \big|_{\varepsilon=0}$ is proportional to $\frac{1}{n} \sum_{i=1}^n D_{P'}(O_i)$ up to matrix scaling by $\Lambda_{P'}$. Starting with an initial estimator $P_{n,0}$, we obtain the MLE $\varepsilon_n^* \in \argmax_{\varepsilon \in \mathbb{R}^s} L_n(P_{n,0}(\varepsilon))$ by fluctuating the initial estimator. If $\varepsilon_n^{*}$ happened to be the zero vector, then the first-order equations characterising the MLEs $\varepsilon_n^*$ and $P_n^* := P_{n,0}(\varepsilon_n^*)$ would imply that Equation \eqref{eqn::tmle1score} is satisfied, as we desire. Although $\varepsilon_n^*$ is typically not zero after one update step, the first-order equations characterising the MLE still imply that 
$$\frac{1}{n}\sum_{i=1}^n \Delta_i H_{P_{n,0}}(A_i, W_i, T_i) [J_i -  \mu_{P_{n}^*}(A_i,W_i, T_i)] =0,$$
We can construct a targeted MLE $P_{n}^*$ that satisfies Equation \eqref{eqn::tmle1score} using an iterative TMLE that is defined as follows. 


The algorithm updates estimators iteratively using maximum likelihood estimation. To begin, we set $k=0$ and initialise $P_{n,0}^{(k=0)} = P_{n,0}$. For each $k \in \mathbb{N}$, we define updated estimators $P_{n,0}^{(k)} := P_{n,0}^{(k-1)}(\varepsilon_n^{(k)})$ where $\hat \varepsilon_n^{(k)} = \argmax_{\varepsilon \in \mathbb{R}^s} L_{n}(P_{n,0}^{(k-1)}(\varepsilon)) $ is the MLE obtained after $k$ iterations. This recursion implies that $\beta(P_{n,0}^{(k)}) = \beta(P_{n,0}^{(k-1)}) + \varepsilon_n^{(k)}$. Standard software can be used to perform iterative multivariate logistic regression to compute these MLEs. We iterate this maximum likelihood update until the updated MLE $\hat \varepsilon_n^{(K)}$ at iteration $K$ is sufficiently close to the zero vector, obtaining an estimator $P_{n}^{*}:= P_{n,0}^{(K)}$ that approximately solves the efficient score equation. The TMLEs for $\beta(P)$ and $OR(P)(w,t)$ are then given by $\beta(P_{n}^{*})$ and $\exp {\underline{f}^T(w,t) \beta(P_{n}^{*}) }$. The procedure should be iterated until the efficient score is solved at a level $\frac{1}{n} \sum_{i=1}^n D_{P_{n}^{*}}(O_i) =  o_P(1 /\sqrt{n})$ for our theoretical results to apply.


\subsection{Semiparametric estimation when pre-and-post-infection variables inform strain missingness} \label{section::spTMLEMissing}

We present asymptotically linear TMLEs for $OR_{adj}(P)$ and $\beta_{adj}(P)$ under the semiparametric assumption that $OR_{adj}(P)(w,t)$ equals $\exp \left\{\beta_{adj}(P)^T\underline{f}(w,t)\right\}$ for some $\beta_{adj}(P) \in \mathbb{R}^s$. The conditional odds ratio $OR_{adj}(P)$ of the observed data distribution identifies the population odds ratio $OR_F(P_F)$ under the missing-data assumptions of Theorem \ref{theorem::ident2}. As such, under the identification, our semiparametric assumption is equivalent to assuming $\Psi_F(P_F)$ satisfies \eqref{eqn::parametricForm} with $\beta_F(P_F) = \beta_{adj}(P)$. These TMLEs are valid even when post-infection covariates $W_T$ inform the viral strain $J$ and strain missingness $\Delta$. Compared to the TMLEs in the previous section, the proposed TMLEs offer valid inference in a wider range of scenarios.

 To this end, let $\mathcal{M}_{adj}$ be a semiparametric statistical model corresponding with the assumption that each $P' \in \mathcal{M}$ satisfies $\mu_{adj, P'} =  \text{expit}\{ a\cdot \beta_{adj}(P')^T\underline{f}(w,t) + h_{adj,P'}(w,t)\}$ for some $\beta_{adj}(P') \in \mathbb{R}^s$ and $h_{adj, P'}: \mathbb{R}^{d} \times \mathbb{T} \rightarrow \mathbb{R}$. Assuming $P$ belongs to the statistical model $\mathcal{M}$, we can express the adjusted odds ratio $OR_{adj}(P)(w,t)$ as $\beta_{adj}(P)^T\underline{f}(w,t)$. However, deriving the efficient influence function for estimating $\beta_{adj}$ under the semiparametric statistical model $\mathcal{M}_{adj}$ is challenging since the parameter involves marginalisation over $W_T$. Instead, we use techniques for censored-data models developed in van der Laan, Robins (2003) to provide a closed-form influence function for $\beta_{adj}$ that is potentially inefficient. Although this influence function may be inefficient, it reduces to the efficient influence function of $\beta_{adj}$ under $\mathcal{M}_{adj}$ if there is no strain missingness (i.e., $P(\Delta=1)=1$), indicating that it is a reasonable choice in terms of efficiency.

For each $P' \in \mathcal{M}_{adj}$, let $\sigma^2_{adj, P'}(a,t,w) := \mu_{adj, P'}(a,t,w)(1-\mu_{adj, P'}(a,t,w))$ and define the function
$$H_{adj, P'}(a,w,t) :=  \left(a  - \frac{ E_{P'}[ A \sigma^2_{adj, P'}(a=1, t ,w)  \mid  T = t, W = w]
}{ E_P[ \sigma^2_{adj,P'}(A,t,w)\mid  T = t, W = w]}\right).$$
Next, define the weight matrix $\Lambda_{adj,P}$ as
$$\Lambda_{adj,P'}^{-1}:= E_{P'}\left\{\underline{f}(W,T)\underline{f}^T(W,T)\frac{(1-\pi_{P'}(W,T))\pi_{P'}(W,T) \sigma^2_{adj, P'}(1,W,T)\sigma^2_{adj, P'}(0,W,T)}{  (1-\pi_{P'}(W,T))\sigma^2_{adj, P'}(0,W,T) +\pi_{P'}(W,T)\sigma^2_{adj, P'}(1,W,T) }\right\},$$
where we assume the inverse exists.


\begin{theorem}
 
A vector-valued influence function for the parameter $\beta_{adj}$ at $P$ under $\mathcal{M}_{adj}$ is
\begin{align*}
    D_{adj,P}(o) :=   &  \frac{\delta}{\Pi_P(w_T,t, a,w)} \odot \left\{\Lambda_{adj,P} \underline{f}(w,t) \right\} \odot   H_{adj, P}(a,w,t)  \left[ j  -  \overline{\mu}_P(w_T,t,a,w) \right] \\
& +   \left\{ \Lambda_{adj,P} \underline{f}(w,t) \right\} \odot H_{adj,P}(a,w,t)  \left[\overline{\mu}_P(w_T,t,a,w)  -  {\mu}_{adj,P}(a,w,t) \right],
\end{align*}
where the product ``$\odot $" is taken coordinate-wise.
\label{theorem::EIF::informative}
\end{theorem}
Using the same approach as in the previous section, the above influence function can be used to construct a one-step estimator for $\beta_{adj}(P)$. We now give a TMLE for $\beta_{adj}(P)$ and $OR_{adj}(P)$, which, although more involved, is similar to the TMLE of the previous section. Let $P' \in \mathcal{M}_{adj}$ be an arbitrary offset. We define the following (least-favorable) fluctuation submodel and working log-likelihood function:

\begin{itemize}
    \item \textit{Logistic fluctuation submodel:} $(P'(\varepsilon): \varepsilon = (\varepsilon_1, \varepsilon_2) \in \mathbb{R}^s \times \mathbb{R}^s) \subset \mathcal{M}$ is a submodel that only fluctuates $\mu_{adj.P'}$ and $\overline{\mu}_{P'}$ through the paths $(\mu_{adj,P'}(\varepsilon_1): \varepsilon_1 \in \mathbb{R}^s)$ and $(\overline{\mu}_{P'}(\varepsilon_1): \varepsilon_1 \in \mathbb{R}^s)$ where 
    \begin{align*}
       \text{logit}\left\{\mu_{adj,P'}(\varepsilon_1)(a,w,t) \right\} &:= \text{logit}\left\{\mu_{adj, P'}(a,w,t) \right\} + \varepsilon_1^T \underline{f}(w,t) H_{P'}(a,w,t): \varepsilon_1 \in \mathbb{R}^s; \\
        \text{logit}\left\{\overline{\mu}_{P'}(\varepsilon_2)(w_T, t, a,w) \right\} &:= \text{logit}\left\{\overline{\mu}_{P'}(w_T, t, a,w) \right\} + \varepsilon_2^T \underline{f}(w,t) H_{P'}(a,w,t): \varepsilon_2 \in \mathbb{R}^s.
    \end{align*}  
    \item \textit{Working log-likelihood function:} $\varepsilon \mapsto L_n(P'(\varepsilon)) :=L_{n,1}(\varepsilon_1;\, P')) +L_{n,2}(\varepsilon_2;\, P') $ where
    \begin{align*} 
  L_{n,1}(\varepsilon_1;\, P') &:=  \frac{1}{n}\sum_{i=1}^n \left\{ \overline{\mu}_{P'}(X_i) \log \mu_{adj,P'}(\varepsilon_1)(Z_i) + (1- \overline{\mu}_{P'}(X_i)) \log(1 - \mu_{adj,P'}(\varepsilon_1)(Z_i)) \right\}; \\
  L_{n,2}(\varepsilon_1;\, P') &:= \frac{1}{n}\sum_{i=1}^n  \frac{\Delta_i}{\Pi_{P'}(X_i)}\left\{J_i \log \overline{\mu}_{P'}(\varepsilon_2)(X_i) + (1-J_i) \log(1 - \overline{\mu}_{P'}(\varepsilon_2)(X_i)) \right\},
    \end{align*}
    $\text{ and } X_i := (W_{T,i}, T_i, A_i, W_i),\, Z_i := (  A_i, W_i, T_i)$.
\end{itemize}

It can be verified that $\frac{d}{d\varepsilon} L_n(P'(\varepsilon)) \mid_{\varepsilon = 0}$ is proportional to the score $P_n D_{adj, P'}$ up to matrix scaling by $\Lambda_{adj, P'}$. The working log-likelihood function, $\varepsilon \mapsto L_n(P'(\varepsilon))$, is a sum of two log-likelihood factors. $L_{n,2}(\varepsilon_2; ,P')$ is a weighted log-likelihood function based on the inverse probability of missingness and only depends on the fluctuation model for $\overline{\mu}_{P'}$. Meanwhile, $L_{n,1}(\varepsilon_1; ,P')$ is a log-likelihood function that only depends on the fluctuation model for $\mu_{adj,P'}$ and involves a pseudo-outcome based on $\overline{\mu}_{P'}$. Since $\varepsilon_1$ and $\varepsilon_2$ are variation independent, the MLE $\varepsilon_n^* = \argmax_{\varepsilon \in \mathbb{R}^s \times \mathbb{R}^s }  L_n(P'(\varepsilon))$ equals $(\varepsilon_{n,1}^*, \varepsilon_{n,2}^*)$ where $\varepsilon_{n,j}^* = \argmax_{\varepsilon_j \in \mathbb{R}^s  }  L_{n,j}(\varepsilon_j ; P')$ for $j \in \{1,2\}$. Thus, the offset MLE can be computed using standard software for weighted logistic regression.

  Let $P_{n,0} \in \mathcal{M}_{adj}$ be an initial estimator of $P_0$, which can be obtained using the methods described in Appendix \ref{section::pluginestimator} of the Supplementary Information. As in the previous section, we can construct a targeted estimator $P_{n}^*$ such that $P_n D_{adj, P_{n}^*} = 0$ by iterating weighted multivariate logistic regression until convergence in an appropriate sense. To this end, let $k \in \{0\} \cup \mathbb{N}$ index the maximum likelihood iteration. Initialise $P_{n,0}^{(k = 0)} := P_{n,0}$ and, for $k \in \mathbb{N}$, recursively define $P_{n,0}^{(k)} := P_{n,0}^{(k-1)}(\varepsilon_n^{(k)})$ where $\varepsilon_n^{(k)} = \argmax_{\varepsilon \in \mathbb{R}^s \times \mathbb{R}^s} L_n(P_{n,0}^{(k-1)}(\varepsilon_n^{(k)}))$. Continue this recursion $K \in \mathbb{N}$ times where $K$ is such that $P_n D_{adj, P_{n,0}^{(K)}} = o_P(n^{-1/2})$. Letting $P_{n}^* =: P_{n,0}^{(K)}$, the TMLEs of $\beta_{adj}(P)$ and $OR_{adj}(P)(w,t)$ are then given by $\beta_{adj}(P_{n}^*)$ and $\exp \left\{ \beta_{adj}(P_{n}^*)^T \underline{f}(w,t)\right\}$.

\subsection{Inference for the conditional odds ratios}
The following Theorems \ref{theorem::limitspTMLE} and \ref{theorem::limitspTMLEMissing} characterise the asymptotic distribution of the TMLE coefficient estimators introduced in Sections \ref{section::spTMLE} and  \ref{section::spTMLEMissing} By applying the delta-method, we obtain the limiting distributions of the TMLEs for $OR(P)$ and $OR_{adj}(P)$ as well.

Let $P_{n}^* \in \mathcal{M}$ be the targeted distribution for $\beta(P)$ introduced in Section \ref{section::spTMLE} obtained from an initial estimator $P_{n,0} \in \mathcal{M}$. Denote $\widetilde \pi_n := \widetilde \pi_{P_{n,0}}$ and ${\mu}_{n}^* := {\mu}_{P_{n}^*}$. The following theorem establishes conditions under which $\beta(P_{n}^*)$ is an asymptotically linear and efficient estimator of $\beta(P)$.

\begin{condition}
   $\underline{f}(W,T)$ is $P_F$-uniformly bounded, and $E_{P}[\underline{f}(W,T)\underline{f}(W,T)^T]$ is invertible.
    \label{cond::limitFeature}
\end{condition}

\begin{condition}
 \textit{Boundedness:} There exists some $\delta > 0$ such that $P(1-\delta > \mu_{n}^*(A,W,T) > \delta) = P(1- \delta > \mu_{P}(A,W,T) > \delta) = 1$.
 \label{cond::limitspTMLEaBounded}
\end{condition}
\begin{condition}
 \textit{Estimators fall in Donsker class:} $\widetilde{\pi}_{n}$ and $\mu_{n}^*$ fall with probability one in a uniformly bounded $P$-Donsker function class. \label{cond::limitspTMLEb}
\end{condition}
\begin{condition}
\textit{Sufficient nuisance rates:} $\|\widetilde{\pi}_{n} - \widetilde{\pi}_P\| = o_P(n^{-1/4})$ and $\max_{a \in \{0,1\}}\|\mu_{n}^*(a, \cdot) - \mu_{P}(a,\cdot)\| = o_P(n^{-1/4})$. \label{cond::limitspTMLEc}
\end{condition}

\begin{theorem}
Assume that $P \in \mathcal{M}$, the conditions of Theorem \ref{theorem::ident} hold, and that conditions  \ref{cond::limitFeature}-\ref{cond::limitspTMLEc} hold. Then,
$$\sqrt{n} \left(\beta(P_{n}^*) - \beta(P) \right) \longrightarrow_d N\left(0, \text{cov}(D_P) \right),$$
where the convergence is jointly in distribution as a random vector and $\text{cov}(D_{P}(\cdot)) \in \mathbb{R}^s \times \mathbb{R}^s$ is the covariance matrix $\text{cov}(D_P) := E_P\left[D_{P}(O)D_{P}^{T}(O) \right]$.
\label{theorem::limitspTMLE}
\end{theorem}

 The following corollary follows from an application of the multivariate delta-method and allows for the construction of confidence intervals for the relative conditional vaccine efficacy function. Under the conditions of the previous theorems, the limiting covariance matrix can be consistently estimated by the empirical covariance of the efficient influence function at the estimated data-generating distribution.  
 
 \begin{corollary}
Under the conditions of Theorem \ref{theorem::limitspTMLE}, we have the following pointwise limit distribution,
 $$\sqrt{n}\left\{ \log OR(P_n^*)(w,t) - \log OR(P)(w,t)\right\} \longrightarrow_d N\left(0 ,  \underline{f}^T(w,t)\text{cov}(D_{P})\underline{f}(w,t) \right).$$
 \end{corollary}

Since $\beta(P_{n}^*)$ is asymptotically linear with influence function being the efficient influence function, it follows from Bickel et al. (1993) that $\beta(P_{n}^*)$ is a locally efficient estimator. The first part of Condition \ref{cond::limitFeature} imposes uniform bounds on the feature mapping $\underline{f}$, and can potentially be relaxed with a more careful analysis. The second part of Condition  \ref{cond::limitFeature} ensures that the coefficient vector $\beta(P)$ that satisfies \eqref{eqn::parametricForm} can be uniquely identified from $P$. Condition \ref{cond::limitspTMLEaBounded} bounds the conditional odds ratios, and its violation can lead to estimator instability. Condition \ref{cond::limitspTMLEb} requires well-behaved nuisance estimators and can be relaxed to allow for black-box algorithms using sample-splitting or cross-fitting (Schick, 1986; van der Laan, Rose, 2011; Chernozhukov et al., 2018)\nocite{Schick}. Condition \ref{cond::limitspTMLEc} requires the targeted nuisance estimators to converge to true target functions faster than $n^{-1/4}$. Under mild regularity conditions, this condition holds when the untargeted nuisance estimators converge at a rate faster than $n^{-1/4}$ --- see Theorem of \cite{vanderLaanRose2011}. Smoothing splines \citep{friedman1991multivariate, tibshirani1996regression}, reproducing kernel Hilbert space estimators, the highly adaptive lasso \citep{HAL2016}, and some neural network architectures \citep{Farrell2018DeepNN} can satisfy both Conditions \ref{cond::limitspTMLEb} and \ref{cond::limitspTMLEc}.

Next, let $P_{n}^* \in \mathcal{M}$ denote the targeted distribution for $\beta_{adj}(P)$ introduced in Section \ref{section::spTMLEMissing} obtained from an initial estimator $P_{n,0} \in \mathcal{M}_{adj}$ of $P_0$. Denote $\pi_n := \pi_{P_{n,0}}$, $\Pi_n:= \Pi_{P_{n,0}}$, $\overline{\mu}_n^* := \overline{\mu}_{P_{n}^*}$, and ${\mu}_{n, adj}^* := {\mu}_{adj, P_{n}^*}$. The following theorem establishes conditions under which $\beta_{adj}(P_{n}^*)$ is an asymptotically linear estimator of $\beta_{adj}(P)$.

 \begin{condition}
 \textit{Boundedness:} There exists some $\delta > 0$ such that $\mu_{n}^*$, $ \mu_{P}$, $\overline{\mu}_{n}^*$, and $\overline{\mu}_P$ take values in $[\delta, 1-\delta]$ with probability one.
 \label{cond::limitspTMLEMissingposBounded }
\end{condition}
\begin{condition}
\textit{Strong overlap for strain missingness:} There exists some $\delta > 0$ such that $P(1 - \delta> \Pi_{n}(W_T, T, A, W) >   \delta) = P(1 - \delta> \Pi_P(W_T, T, A, W) >   \delta) = 1$.
 \label{cond::limitspTMLEMissingpos} 
\end{condition}
\begin{condition}
 \textit{Estimators fall in Donsker class:} $\pi_{n}$, $\Pi_{n}$,  $\overline{\mu}_{n}^*$ and $\mu_{n,adj}^*$ fall with probability one in a uniformly bounded $P$-Donsker function class.
 \label{cond::limitspTMLEMissinga} 
\end{condition}
\begin{condition}\textit{Sufficient nuisance rates:}
$\|\pi_n - \pi_P\| = o_P(n^{-1/4})$, $\max_{a \in \{0,1\}}\|\mu_{n,adj}^*(a, \cdot) - \mu_{adj,P}(a,\cdot)\| = o_P(n^{-1/4})$, $\|\Pi_n - \Pi_P\| + \|\overline{\mu}_n^* - \overline{\mu}_P\| = o_P(1)$ and  $\|\Pi_n - \Pi_P\|\|\overline{\mu}_n^* - \overline{\mu}_P\|= o_P(n^{-1/2})$. \label{cond::limitspTMLEMissingb}
\end{condition}

\begin{theorem}
Assume that $P \in \mathcal{M}_{adj}$ and Conditions  \ref{cond::limitFeature},  \ref{cond::limitspTMLEMissingposBounded }-\ref{cond::limitspTMLEMissingb} hold. Then,
$$\sqrt{n} \left(\beta_{adj}(P_{n}^*) - \beta_{adj}(P) \right) \longrightarrow_d N\left(0, \text{cov}(D_{adj,P} ) \right),$$
where the convergence is jointly in distribution as a random vector and $\text{cov}(D_{adj,P}) \in \mathbb{R}^{s\times s} $ is the covariance matrix $\text{cov}(D_{adj,P}(\cdot)) := E_P\left[D_{adj,P}(O)D_{adj,P}^{T}(O) \right]$.
\label{theorem::limitspTMLEMissing}
\end{theorem}

 \begin{corollary}
Under the conditions of Theorem \ref{theorem::limitspTMLEMissing}, we have the following pointwise limit distribution,
 $$\sqrt{n}\left\{ \log OR_{adj}(P_n^*)(w,t) - \log OR_{adj}(P)(w,t)\right\} \longrightarrow_d N\left(0 ,  \underline{f}^T(w,t)\text{cov}(D_{adj,P})\underline{f}(w,t) \right).$$
 \end{corollary}

Conditions \ref{cond::limitspTMLEMissingposBounded } and \ref{cond::limitspTMLEMissingpos} impose mild conditions on the estimators and data-generating distribution. Condition \ref{cond::limitspTMLEMissinga} can also be relaxed using sample-splitting or cross-fitting. Condition \ref{cond::limitspTMLEMissingb} is similar to the rate conditions of Condition \ref{cond::limitspTMLEc}. Notably, the estimator is doubly robust in the nuisance estimation rates of $\Pi_n$ and $\overline{\mu}_n$. Consequently, as long as $\Pi_n$ is estimated sufficiently fast, the TMLE can remain asymptotically normal even when $\overline{\mu}_n$ is estimated at rates slower than $n^{-1/4}$.

 
 \section{Simulations}

  We assess the performance of the TMLE for the conditional odds ratio $OR_F(P_F)(w,t)$ when there is strain-missingness induced by covariates that are causally affected by vaccination and time-of-infection. We assess the performance of both the TMLE that adjusts for post-infection-informed missingness and the generally-biased TMLE that does not fully adjust for the informative missingness.

We generate the data-structure $(W, A, T, W_T, \Delta, \Delta J)$ as follows. $W = (W_1,W_2)$ is a multivariate truncated normal random variable with bounds $(-1,1)$ and variance vector $(1,1)$ and correlation $-0.5$. We generate the treatment as $A \mid W_1,W_2 \sim \text{Bernoulli}(\text{expit}(0.75\cdot(W_1 + W_2))$, the time-to-event variable as $T \mid W_1, W_2,A \sim \text{Weibull}(\text{shape} =3, \, \text{scale} = 1/\exp\left\{ 0.1*(W_1 + W_2 - 0.5)\right\}) $, and the post-infection covariate
  as $W_T \mid T, A, W \sim \text{Bernoulli}(p_{W_T})$ where $p_{W_T} := 0.1 + 0.85\cdot \text{expit}\left\{-1 + W_1/2 + W_2/2 + 2.5 \cdot A + A \cdot(W_1 - W_2)/4 - 0.5 \right\}.$
  We choose the conditional distribution of the strain $J \mid W, A, T $ marginalized over $W_T$ such that 
  $P_F(J=1 \mid T, \, A,\, W) = \text{expit}(-0.75 + A*(\beta_{0,F} + \beta_{1,F}(T-1)) + (W_1 + W_2 + T-1)/2)$
  where $\beta_{0,F} = 1/2$ and $\beta_{1,F} = 1/2$ are the population target estimands of interest. By doing this, we guarantee that the conditional distribution of $J \mid W, A, T$ satisfies the partially linear logistic regression model assumption, which is necessary for consistency of our methods.
 To do this, we set $ P_F(J=1 \mid T, \, A,\, W,\, W_T=0) =\text{expit}(-1.3 + A \cdot(0.3 + (T-1)/6 + W_1/2) (W_1 + W_2 + T -1)/2)$ and determine $P_F(J=1 \mid T, \, A,\, W,\, W_T=1)$ by inverting the identity $P_F(J=1 \mid T, \, A,\, W) = P_F(J=1 \mid T, \, A,\, W, \, \mid W_T = 1)p_{W_T}+  P_F(J=1 \mid T, \, A,\, W, \, \mid W_T = 0)(1-p_{W_t})$. Finally, we generate the strain-missingness as $\Delta \sim \text{Bernoulli}(0.05  + 0.95\cdot \text{expit}(1.5 - 2W_T + (W_1 + W_2 + T-1)/2) - A - A \cdot W_T/2)$. For sample sizes $n = 100, 150,250, 500, 1000, 1500, 2000$ and $2500$, we generated 1000 case-only datasets. For each dataset, approximately half of the viral strain measurements were missing. We then computed the Monte-Carlo bias, standard error, mean squared error, and confidence interval coverage of the following four estimation methods for estimating $\beta_F := (\beta_{0,F}, \beta_{1,F})$.
 
 We compare four estimators for $\beta(P)$ and $\beta_{adj}(P)$ in our study. The first two are TMLE ($\beta$) and TMLE ($\beta_{adj}$), described in Sections \ref{section::spTMLE} and \ref{section::spTMLEMissing}, respectively. The third estimator is glmNaive ($\beta_{adj}$), which estimates $\beta_{0,F}$ and $\beta_{1,F}$ using interaction coefficients in a misspecified logistic-regression model for $J \mid W_T, T, A, W, \Delta = 1$. This estimator is naive since it adjusts for the post-treatment variables (i.e., bad controls) directly within the logistic regression model \citep{cinelli2020crash}. The fourth estimator is glm $(\beta)$, which estimates $\beta_{0,F}$ and $\beta_{1,F}$ using interaction coefficients in a misspecified logistic-regression model for $J \mid A, T, W, \Delta = 1$. The logistic regression model of the third estimator includes all main-terms of $(W, W_T, T, A)$ and all interactions between $A$ and $W_T$, and the model of the fourth estimator includes all main-term and two-way interactions of the variables $(W, T, A)$. We provide empirical bias, standard error, root-mean-squared error, and 95\% confidence interval coverage for $\beta_0$ and $\beta_1$ for all estimators in Figures \ref{fig:inter} and \ref{fig:cov}.
 
 
 \begin{figure}[H]
     \centering
    \begin{subfigure}{0.45\linewidth}
     \includegraphics[width =  \linewidth]{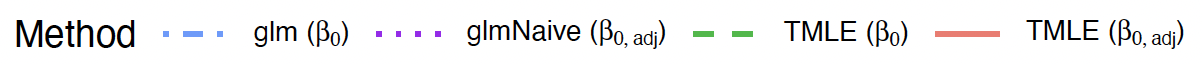}
      \includegraphics[width = 0.5 \linewidth]{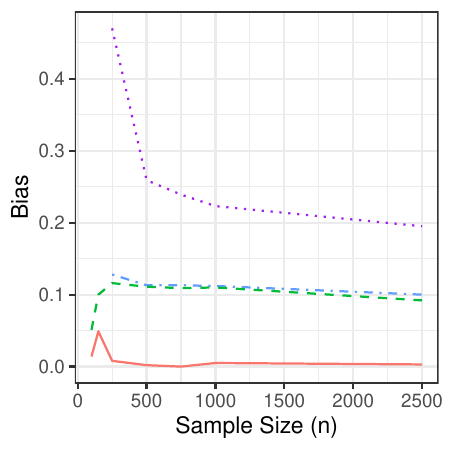}\includegraphics[width = 0.5 \linewidth]{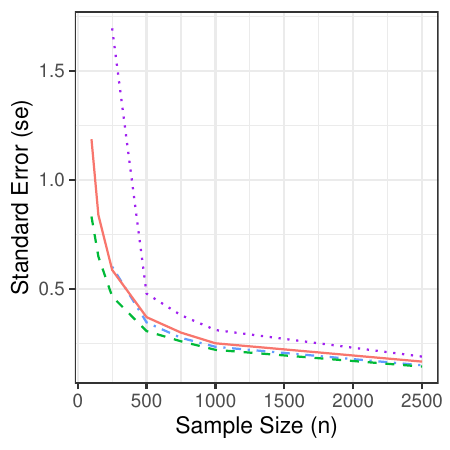}
      \includegraphics[width = 0.5 \linewidth]{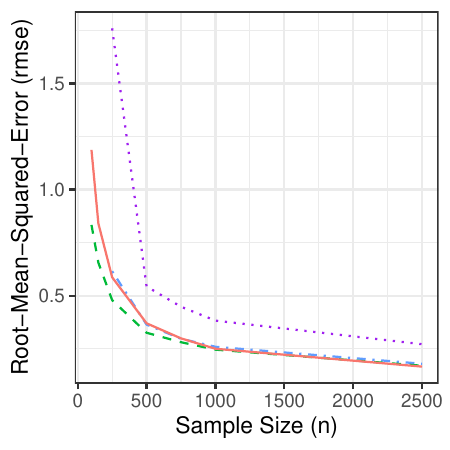}\includegraphics[width = 0.5 \linewidth]{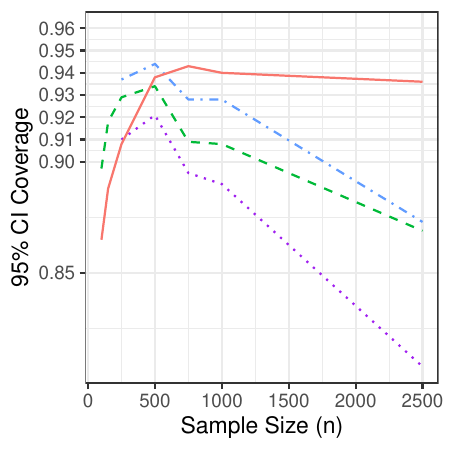}
     \caption{Evaluation of TMLEs for $\beta_{0,F}$.}
      \label{fig:inter}
  \end{subfigure}\begin{subfigure}{0.45\linewidth}
  \includegraphics[width =  \linewidth]{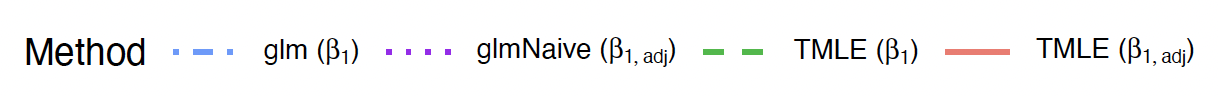}
      \includegraphics[width = 0.5 \linewidth]{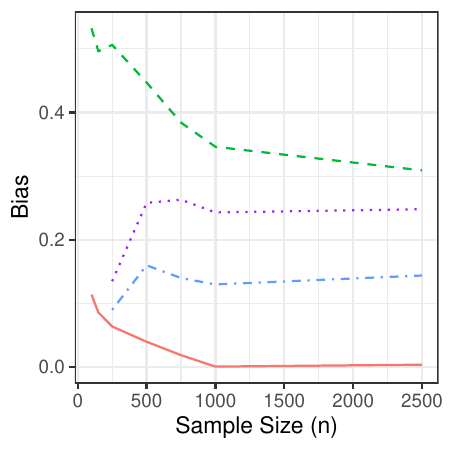}\includegraphics[width = 0.5 \linewidth]{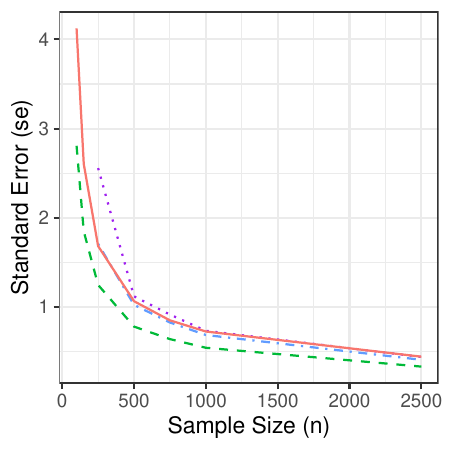}
      \includegraphics[width = 0.5 \linewidth]{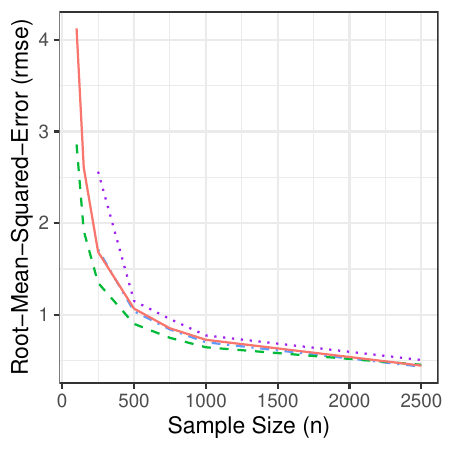}\includegraphics[width = 0.5 \linewidth]{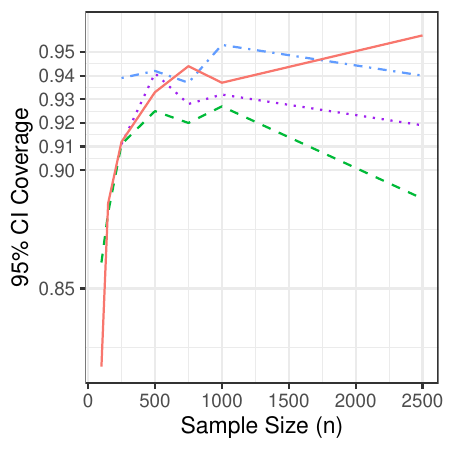}
      \caption{Evaluation of TMLEs for $\beta_{1,F}$.}
      \label{fig:cov}
  \end{subfigure}
   \caption{Empirical bias, standard error, root-mean-squared-error and 95\% confidence coverage for various estimators of the intercept and coefficient of time variable for the conditional odds ratio based on 1000 Monte-Carlo simulations.}
  \end{figure}

 Figures \ref{fig:inter} and \ref{fig:cov} demonstrate that TMLE $(\beta_{adj})$ is the only estimator that appears to be unbiased when adjusting for post-infection informed missingness. However, this better performance comes at the cost of a larger standard error, as expected by theory. On the other hand, glmNaive $(\beta_{adj})$, which naively adjusts for post-treatment variables, performs poorly overall and should not be used in practice. The TMLE $(\beta)$ performs best in root-mean-squared error and standard error, despite its semiparametric model being misspecified. However, TMLE $(\beta_{adj})$ performs significantly better than TMLE $(\beta)$ in bias and 95\% confidence interval coverage, especially for larger sample sizes. In practice, choosing between TMLE $(\beta_{adj})$ and TMLE $(\beta)$ involves a trade-off between bias due to post-treatment-induced confounding and estimator standard error. Regarding confidence interval coverage, TMLE $(\beta_{adj})$ may generally be preferred. The misspecified estimator glm $(\beta)$ performs reasonably well in MSE relative to TMLE $(\beta_{adj})$, but it had poor confidence interval coverage in larger samples. Furthermore, for sample sizes smaller than 250, glm $(\beta)$ was ill-defined because the number of interaction variables was too large. Notably, TMLE $(\beta)$ and TMLE $(\beta_{adj})$ have the benefit of allowing for the use of high dimensional machine-learning algorithms like LASSO and thus are applicable in such settings where ordinary logistic regression fails. Interestingly, TMLE $(\beta)$ outperforms glm $(\beta)$ in standard error, suggesting that TMLE $(\beta)$ may sometimes be more efficient than glm $(\beta)$ when the latter estimator is misspecified.

 \section{Application}
 We apply the newly developed TMLEs 
  to the ENSEMBLE randomized, placebo-controlled COVID-19 vaccine efficacy trial in the U.S., South Africa, and Latin America (Sadoff et al., 2022).\nocite{sadoff2022final}  The primary objective of ENSEMBLE was to assess the efficacy of the Ad26.CoV2.S vaccine to reduce the rate of virologically-confirmed moderate-to-severe COVID-19 occurring at least 14 days and at least 28 days after a single dose of the vaccine or placebo, in people not previously infected with SARS-CoV-2 as documented by negative results of diagnostic tests applied to blood samples drawn at enrollment. 
 At Latin American study sites, several viral strains (``lineages") of the SARS-CoV-2 virus circulated and caused moderate-to-severe COVID-19 endpoints at substantial prevalence: the original Wuhan/ancestral lineage from which the SARS-CoV-2 reference strain was derived for engineering into the Ad26.CoV2.S vaccine construct, and four variant strains that emerged (gamma, lambda, mu, zeta) (see Figures 1 and 3 of Sadoff et al., 2022, where we use the term ``ancestral lineage" to denote all strains close genetically to the vaccine-strain, comprising the ``reference strain", ``other", and ``other + E484K" in the Sadoff et al. nomenclature).  Another objective of the ENSEMBLE study focusing on Latin America is to compare vaccine efficacy between the ancestral lineage and each of the variants, to understand whether and how much vaccine efficacy was abrogated by emerging variants. As a randomized controlled trial, this objective could be assessed based on survival analysis methods, as done in Sadoff et al. (2022)\nocite{sadoff2022final} and Magaret et al (2024)\nocite{magaret2024quantifying}.  Here, we analyze the same data set, coarsened only to include the information on outcome cases (moderate-to-severe COVID-19).  In so doing, we pretended that the study was a case-only observational study, which usefully provides the opportunity to compare the case-only results to results obtained based on the full randomized controlled trial data. For our analysis, we restrict to studying the relative conditional vaccine efficacy in Latin America, since the variants gamma, lambda, mu, and zeta only circulated in Latin America, not at the study sites in South Africa and the United States. We focus on the study endpoint occurring at least 14 days after a single dose of the vaccine or placebo.   
 
 In applying the TMLE methods, we make the assumption that the conditional odds ratio/relative strain-specific conditional vaccine efficacy is constant (i.e. an intercept model in the partially linear logistic regression) and report estimates of this constant parameter. To account for confounding due to informative lineage missingness and improve efficiency, we adjust for the following baseline variables: participant age, 
 sex, race, country, baseline risk score, time of COVID-19 endpoint since the first person enrolled (September 21, 2020), and calender time at enrollment. For the zeta variant, we do not adjust for country since zeta only appeared in some Latin American countries. The baseline risk score  is the same as used in \cite{FongRiskScoreEnsemble} that was built from the placebo arm using machine learning techniques. In the context of semiparametric logistic regression, efficiency depends on the choice of adjustment variables $W$. Adjusting for \( W \) correlated with \( J \) but not \( A \) improves efficiency, while adjusting for \( W \) correlated with \( A \) but not \( J \) can reduce precision without reducing bias. This follows standard confounding adjustment principles: adjust for predictors of the outcome but not for predictors of the exposure that do not influence the outcome \citep{vanderweele2019principles}. We also adjust for the post-treatment variable the SARS-CoV-2 viral load measured from a nasal swab sample taken at the time of detection of COVID-19. In semiparametric logistic regression, efficiency depends on the choice of adjustment variables \( W \). To ensure an unbiased estimate, \( W \) should be correlated with both the treatment \( A \) and the mark \( J \).   

  \begin{figure}[H]
     \centering
     \includegraphics[width=0.5\linewidth]{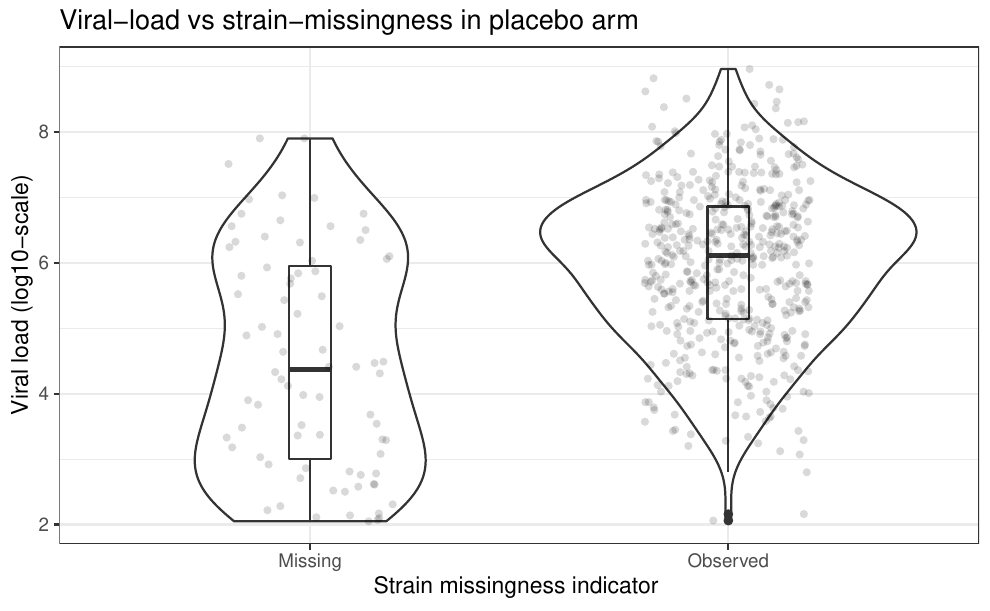}\includegraphics[width=0.5\linewidth]{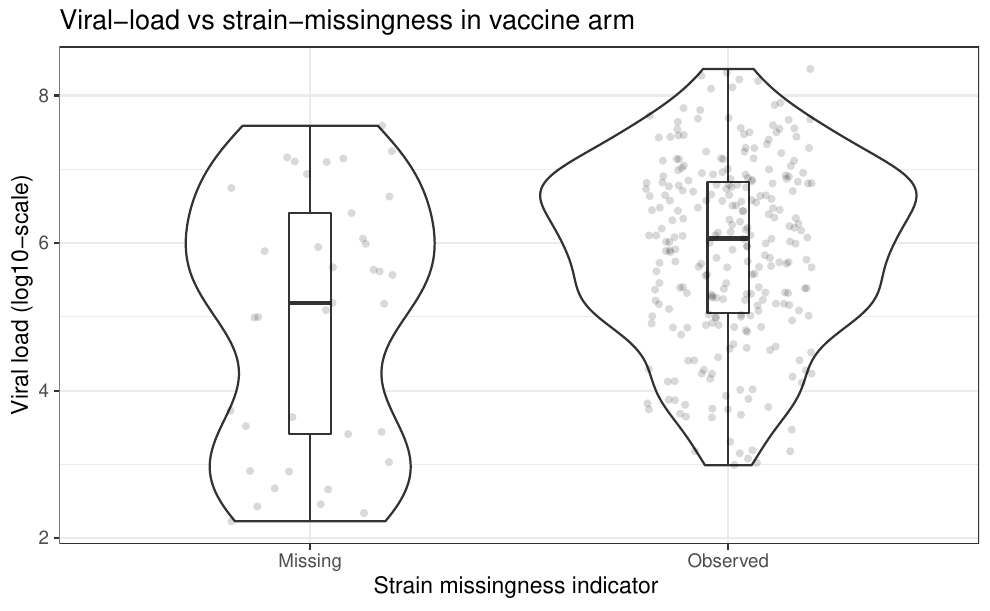}
     \caption{Box and Violin plots of SARS-CoV-2 viral load stratified by vaccination status and the indicator of the infection-causing strain/lineage being observed: ENSEMBLE study Latin America.}
     \label{fig:Viralloadmissing}
 \end{figure}
 
   \begin{figure}[H]
     \centering
     \includegraphics[width=0.7\linewidth]{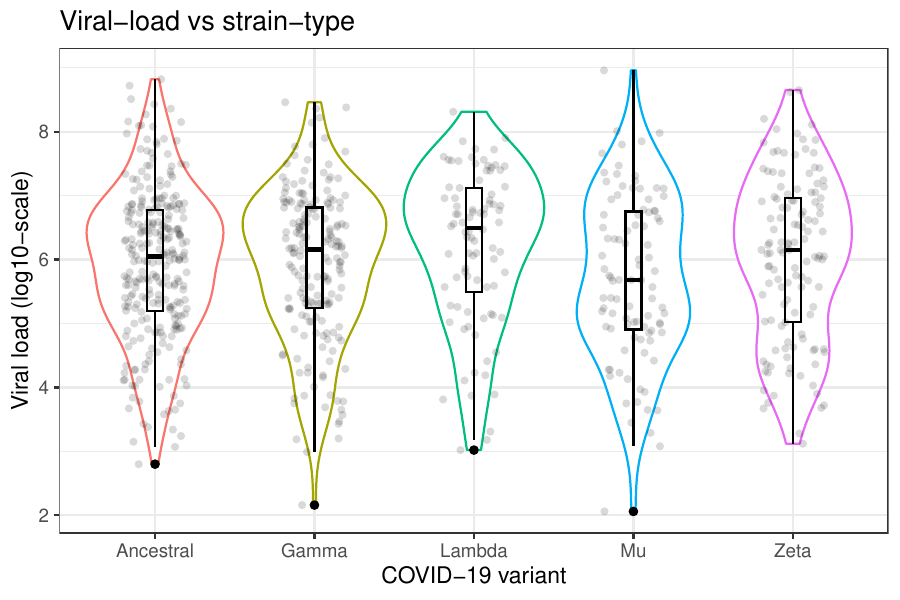}
     \caption{Box and Violin plots of SARS-CoV-2 viral load stratified by the indicator of the infection-causing strain/lineage among cases with an observed strain: ENSEMBLE study Latin America.}
     \label{fig:Viralloadcorrelation}
 \end{figure}

 Figures \ref{fig:Viralloadmissing} and \ref{fig:Viralloadcorrelation} provide box and violin plots showing how the distribution of SARS-CoV-2 viral load at COVID-19 endpoint detection is modified by the strain-missingness indicator ($\Delta$) and by the SARS-CoV-2 lineage among vaccinated and unvaccinated individuals. The distribution of viral load does not appear to differ much between the vaccine and placebo arms. Viral load tends to be lower among individuals whose strain type is missing (Figure \ref{fig:Viralloadmissing}). This is expected since the amplification technology for measuring the SARS-CoV-2 sequence is more likely to succeed if there is more SARS-CoV-2 viral material in the sample. On the other hand, we observe that the distribution of viral load is similar between individuals infected by the ancestral strain compared to with non-ancestral strains (Figure \ref{fig:Viralloadcorrelation}). This suggests that viral load, although highly predictive of strain-missingness, is not very predictive of the strain type and therefore does not lead to significant confounding bias in this study.

 
\begin{center}
\begin{table}[H] \centering
\caption{TMLE estimates of the conditional odds ratio [ = relative strain-specific conditional vaccine efficacy against the ancestral COVID-19 strain ($J=1$) compared to against a variant strain or strains ($J=0$)]: ENSEMBLE study Latin America}
\label{sieveresults}
{\small
\begin{tabular}{lccccccccc} \hline \hline
Strain  & \multicolumn{3}{c}{TMLE w/ inform. missingness} &
\multicolumn{3}{c}{TMLE w/o inform. missingness} & \multicolumn{3}{c}{Failure time analysis$^1$}   \\ 
Comparison$^2$ & Rel. VE & 95\% CI & P-value & Rel. VE & 95\% CI & P-value & Rel. VE & 95\% CI & P-value \\ \hline 
Anc vs else & 0.56 & (0.38, 0.81) & 0.002 & 0.54 & (0.35, 0.83) & 0.005 & 0.62 & (0.46, 0.85)  & 0.0036 \\     
Anc vs gamma    & 0.48 & (0.30, 0.78) & 0.003 & 0.50 & (0.29, 0.87) & 0.015 &  0.58 & (0.39, 0.85)  & 0.0062 \\
Anc vs lambda   & 0.26 & (0.15, 0.48) & $<$ 0.001 & 0.28 & (0.13, 0.59) & 0.001 & 0.39 & (0.24, 0.64)  & $<$ 0.001 \\
Anc vs mu        &  0.61 & (0.31, 1.2) & 0.162 & 0.57 & (0.24, 1.3) & 0.188 & 0.58 & (0.36, 0.94)  & 0.027 \\
Anc vs zeta     &  0.95 & (0.59, 1.54) & 0.83 & 0.96 & (0.56, 1.16) & 0.87  & 1.01 & (0.63, 1.61)  & 0.96 \\ \hline 
\end{tabular}
\newline
\noindent $^1$Competing risks Cox modeling accommodating viral-load-dependent-strain-missingness by augmented inverse probability weighting (Heng et al., 2025)\nocite{Hengetal2025} \newline
\noindent $^2$Number of COVID-19 infection endpoints in the (vaccine, placebo) arm by lineage: ancestral (72, 196), gamma (73, 111), lambda (43, 45), mu (38, 57), zeta (33, 92), non-ancestral (187, 305)
}
\end{table}
\end{center}

 Table \ref{sieveresults} gives estimates of relative strain-specific conditional vaccine efficacy (relVE) against moderate-to-severe COVID-19 for the ancestral strain ($J=1$) relative to both individual and pooled variant strains ($J=0$). Uncertainty in the estimate is assessed through 95\% confidence intervals and  p-values for the null hypothesis that the strain-specific conditional vaccine efficacy is identical for the two strain types being compared. We note that the following relative conditional vaccine efficacy interpretations do not necessarily carry forward to other endpoints such as severe COVID-19 disease. We observe that the relVE point estimates are below one for all strain comparisons. 
 As measured by relVE$\times$ 100\%, for the pooled-variants analysis the level of vaccine protection against non-ancestral strains is 
 estimated to be 56\% (95\% CI: 38\%, 81\%) of that achieved against the ancestral strain. For the individual variant strains gamma, lambda, mu, and zeta, respectively, 
 the level of vaccine protection against the variant is estimated to be 48\% (30\%, 78\%), 26\% (15\%, 48\%), 
 61\% (31\%, 120\%), and 95\% (59\%, 154\%) of that achieved against the ancestral strain.  These results suggest that in general vaccine efficacy for moderate-to-severe COVID-19 is modestly better against the ancestral strain than against variants (p=0.002),
 with vaccine efficacy most abrogated against the lambda variant (p $<$ 0.001) and more modestly against the gamma variant (p = 0.003), and there is inconclusive evidence for any efficacy abrogation against the mu (p=0.162) and zeta (p=0.83) variants.

For comparison we also provide estimates of the relative strain-specific conditional vaccine efficacy using the TMLE of Section \ref{section::spTMLE}, which does not adjust fully for informative missingness. From Table \ref{sieveresults},
we observe that the estimates using both methods are very similar, which supports our earlier claim that viral load is not a significant confounder. In addition, we observe that the confidence intervals for the TMLE of Section \ref{section::spTMLE} are wider than those for the novel TMLE of Section \ref{section::spTMLEMissing}, which suggests that the latter method, although having a negligible bias-reduction, does lead to gains in efficiency.  Lastly, the results based on the full survival analysis data set (Table \ref{sieveresults}) yields comparable point estimates of relative vaccine efficacy, with narrower 95\% CIs and lower p-values, which is expected given the additional information in the data set.

\section*{Acknowledgments} 
Research reported in this publication was supported by the National Institute Of Allergy And Infectious
Diseases of the National Institutes of Health under Award Number R37AI054165 and 
the U.S. Public Health Service Grant AI068635. 
The content is
solely the responsibility of the authors and does not necessarily represent the official views of
the National Institutes of Health.  We thank the ENSEMBLE study participants, study team, and sponsors, especially Janssen statistics including Sanne Roels and An Vandebosch, and thank Li Li and Michal Juraska for survival analysis of ENSEMBLE. 

\newpage

\bibliography{ref}

  \appendix

 \section{Constructing semiparametric substitution estimators for the conditional odds ratio}
 \label{section::pluginestimator}
 In this section, we describe how to obtain substitution estimators for the conditional odds ratio parameters using machine-learning algorithms under the statistical models $\mathcal{M}$ and $ \mathcal{M}_{adj}$. For both models, the method will require a machine-learning estimator that respects the partially linear logistic-link model constraint.

\subsection{Substitution estimators of the conditional odds ratio under $\mathcal{M}$}

We begin with the simplest identification result, given in Theorem \ref{theorem::ident}. In particular, we aim to estimate $\beta$ such that $ \mu_P(a,w,t) = \text{expit}\left\{ a\beta^T \underline{f}(w,t) + h_{P}(w,t) \right\}$, which under causal assumptions identifies $\beta_{F}$ in the expression for 
$OR_F(P_F)$. For a given nuisance function class $\mathcal{H}$ (e.g. the class of functions of bounded variation), define
$(\hat \beta_{n}, \hat h_{n}) := \argmin_{\beta(P)\in \mathbb{R}^s, h \in \mathcal{H}}  \mathcal{R}_n(\beta, h)$
where $\mathcal{R}_n(\beta, h) =$
$$\frac{1}{n}\sum_{i=1}^n R_i \Delta_i \left\{ \log\left\{1 + \exp\left(A_i\beta^T \underline{f}(W_i,T_i) + h(W_i,T_i) \right) \right\} - J_i \left(A_i\beta^T \underline{f}(W_i,T_i) + h(W_i,T_i) \right) \right\}.$$ This risk minimization problem is equivalent to performing the partially-linear logistic regression of $J$ on $(T,A,W)$ using only the observations with $(\Delta=1,R=1)$. This gives rise to an estimator $\hat \mu_{n}(a,w,t) := \text{expit} \left\{  a\hat \beta_{n}^T \underline{f}(w,t) + \hat h_{n}(w,t)\right\} $ of $\mu_P(a,w,t)$, which importantly respects the constraints of the statistical model $\mathcal{M}$. Under the conditions of Theorem \ref{theorem::ident}, the relative conditional vaccine efficacy can then be estimated by the conditional odds ratio substitution estimator,
$$\widehat{OR}_{n}(w,t) := \frac{\hat \mu_{n}(t,1,w)(1-\hat \mu_{n}(t,0,w))}{\hat \mu_{n}(t,0,w)(1-\hat \mu_{n}(t,1,w))} = \exp \left\{ \hat \beta_{n}^T \underline{f}(w,t) \right\}.$$

\subsection{Substitution estimators of the conditional odds ratio under $ \mathcal{M}_{adj}$}

Next, we consider the more general identification result, given by Theorem \ref{theorem::ident2}. For estimation of the relative conditional vaccine efficacy, we utilize a two-stage sequential regression approach. Recall, $ \mu_{adj,P}(a,w,t) =  \text{expit}\left\{ a\beta^T \underline{f}(w,t) + h(w,t) \right\} $. Let $\overline{\mu}_n$ be an arbitrary initial estimator of $\overline{\mu}$ obtained, for instance, using machine-learning algorithms, e.g., gradient-boosting, generalized additive models, the highly adaptive lasso, or ensemble methods like SuperLearner (van der Laan et al., 2007)\nocite{vanderLaanetal2007}. $\overline{\mu}_{n}$ can be obtained by performing the nonparametric (logistic) regression of $J$ on $(T,W_T,A,W)$ using the observed cases with no missing strain types ($R=1, \Delta = 1$). Note, $\overline{\mu}$ is \textit{not} necessarily a partially linear logistic-link model under our statistical model assumptions, and therefore $\overline{\mu}$ should be estimated in a fully nonparametric way. We now utilize sequential regression to obtain an estimator of $\mu_{adj,P} $. Define the pseudo-outcome logistic risk function,
$\mathcal{R}_{n,  \overline{\mu}_n} (\beta, h) = $
$$\sum_{i=1}^n   \left\{\log\left\{1 + \exp\left(A_i\beta^T \underline{f}(W_i,T_i) + h(W_i,T_i) \right) \right\} - \overline{\mu}_n(T_i,W_{T,i}, A_i, W_i) \left(A_i\beta^T \underline{f}(T_i,W_i) + h(T_i,W_i) \right) \right\}.$$
For a given nuisance function class $\mathcal{H}$, define
$(\hat \beta_{n, adj}, \hat h_{n, adj}) := \argmin_{\beta(P)\in \mathbb{R}^s, h \in \mathcal{H}}  \mathcal{R}_{n,  \overline{\mu}_n}(\beta, h).$
In other words, $(\hat \beta_{n, adj}, \hat h_{n, adj})$ is obtained by performing the partially-linear logistic regression of the pseudo-outcome $\overline{\mu}_n(T,W_{T}, A, W)$ on $(T,A,W)$ using only the observations with $R=1$. This gives the estimator $\hat \mu_{n, adj}(a,w,t) := \text{expit} \left\{  a\hat \beta_{n, adj}^T \underline{f}(w,t) + \hat h_{n, adj}(w,t)\right\} $ of $\mu_{adj,P}(a,w,t)$, which importantly respects the constraints of the statistical model $ \mathcal{M}_{adj}$. Under the conditions of Theorem \ref{theorem::ident2}, the relative conditional vaccine efficacy can then be estimated by odds ratio substitution estimator,
$$\widehat{OR}_{n, adj}(w,t) := \frac{\hat \mu_{n, adj}(t,1,w)(1-\hat \mu_{n, adj}(t,0,w))}{\hat \mu_{n, adj}(t,0,w)(1-\hat \mu_{n, adj}(t,1,w))} = \exp \left\{ \hat \beta_{n, adj}^T \underline{f}(w,t) \right\}.$$

\section{Proofs of identification results}

\begin{proof} [Proof of  \textbf{Theorem \ref{theorem::ident}}]

The weak overlap assumptions (\ref{cond::ident::C0}, \ref{cond::ident::C1}, \ref{cond::ident::C2}) ensure that the following conditional expectations are well-defined. By Assumption \ref{cond::ident::C3}, we have $P( J=1|\Delta = 1, R=1, A, W,T) = P( J=1| R=1, A, W,T)$. Therefore, we also have 
$$OR(P)(W,T) = \frac{P( J=1| R=1, A=1, W,T)/P(J=0|R=1, A=1, W,T)}{P( J=1|R=1, A=0, W,T)/P( J=0|R=1,A=1, W,T)}.$$
By Bayes rule and Assumption \ref{cond::ident::C4}, $$P(J=1| R=1, A, W, T=t) = P(J=1| T=t , A, W)  = P(J=1, T =t|  A, W)/ P(T = t| A,W).   $$
A similar result holds for $P(J=0| R=1, A, W, T)$. Substitution into the previous expression for $OR(P)(W)$ and canceling terms  gives
$$OR(P)(W,t) = \frac{P( J=1, T= t|  A=1, W)/P( J=1, T = t| A=0, W) }{P(J=0, T= t|R=1, A=1, W)/P( J=0, T= t|A=0, W)}$$
from which the desired identification result now follows.
 \end{proof}

\begin{proof} [Proof of \textbf{Theorem \ref{theorem::ident2}}]
 The weak overlap assumptions (\ref{cond::ident::C0}, \ref{cond::ident::C2}, \ref{cond::ident::D1}) ensure that the following conditional expectations are well-defined. By Assumptions \ref{cond::ident::C4} and \ref{cond::ident::D3}, we have 
 $$P( J=1| R=1, A=a, W, T) = E\left[P( J=1|R=1, W_T, T, A , W)| R=1, A=a, W, T)\right] $$
 $$=  E\left[P( J=1|\Delta = 1, R=1, W_T, T, A, W)| R=1, A=a, W, T)\right].$$
 The first equality follows from Assumption \ref{cond::ident::C4} and the second equality follows from Assumption \ref{cond::ident::D3}. The remainder of the proof follows exactly as in the proof of Theorem \ref{theorem::ident}.
 
 \end{proof}

 \begin{proof}[Proof of  \textbf{Theorem \ref{theorem::CausalIdent}}]
Under the assumed assumptions, we have
\begin{align*}
  &\hspace{-1cm} P_F(T_t(a) = t, J_t(a) = j \mid T \geq t, A(t-)=0, W) \\
  & =  P_F(T_t(a) = t, J_t(a) = j \mid T \geq t, A(t) = a, A(t-)=0, W) \\
    &= P_F(T = t, J = j \mid T \geq t, A(t) = a, A(t-)=0, W =w)\\
    &= P_F(T = t, J = j \mid T \geq t, A(t) = a, W =w).
\end{align*}
The first equality follows from Assumption \ref{cond::CausalExchange}, the second from Assumption \ref{cond::Causalconsistency}, and the final from Assumption \ref{cond::indepVacTime}. Next, by Markov's inequality, we have
\begin{align*}
    &\hspace{-1cm} P_F(T = t, J = j \mid T \geq t, A(t) = a, W =w) \\
    &=  P_F(  J = j \mid  T  =t, T \geq t, A(t) = a, W =w)P_F(T = t  \mid T \geq t, A(t) = a, W =w).
\end{align*}  
Note that $P_F(T = t  \mid T \geq t, A(t) = a, W =w)$ depends only on the treatment level $a$ and not the strain level $j$. The desired result can be obtained by substituting the above expressions and noting that contributions due to $P_F(T = t  \mid T \geq t, A(t) = 1, W =w)$ and $P_F(T = t  \mid T \geq t, A(t) = 0, W =w)$ cancel in the quotient.

 \end{proof}

\section{Derivations of influence functions}

 \begin{proof}[ \textbf{Proof of Theorem \ref{theorem::EIF::uninformative}}]
    The given efficient influence function is the efficient influence function of the coefficient of partially linear logistic regression model, which is well-studied. Working conditional on $\{R=1, \Delta = 1\}$, the result follows from the derivation given on pages 621-622 and 626-629 of the working paper \cite{OddsRatioreadingsTMLE}.
 \end{proof}

 \begin{proof}[\textbf{Proof of Theorem \ref{theorem::EIF::informative}}]
Without loss of generality, we assume that the missingness identification results of Theorems \ref{theorem::ident} and \ref{theorem::ident2} hold. As these causal assumptions are untestable, we can do this with no loss of generality \citep{vanderlaanunified}. Following \cite{vanderlaanunified}, our proof technique is to map the efficient influence function of $\beta_{adj}(P)$ under a model where all variables are fully observed to an influence function (i.e. gradient) of $\beta_{adj}(P)$ under a model where there is strain and case missingness.

To this end, we first consider the case where we observe the data-structure $(W,A,T,J) \sim P_F$, which has no missingness. By Theorem \ref{theorem::EIF::uninformative}, the function 
$$o \mapsto D_{adj,P}(o) :=   \Lambda_{P} \underline{f}(w,t) \cdot  \delta H_{adj,P}(a,w,t) \left[ j  - \mu_{P}(a,w,t)\right] $$
is the efficient influence function under the semiparametric statistical model that assumes $\mu_{P_F}(W,A,T) := P_F(J=1\mid W, A, T) = \text{expit}\left\{h_{P_F}(W,T) + A \cdot \beta_{adj}(P_F)^T\underline{f}(W,T) \right\}$ and has $P_F(\Delta = 1) = 1$. For ease of notation, we suppress the dependence on the feature vector $\underline{f}$ and denote
$$\widetilde{H}_{adj,P}(a,w,t) := \underline{f}(w,t) \cdot  H_{adj,P}(a,w,t) $$

Since $W_T$ provides no additional information about the parameter $\beta_{adj}(P_F)$ when there is no missingness, we have $D_{adj,P}$ is also that EIF for the data-structure $(W,A,T,W_T, J)$ under the same semiparametric model assumptions. Now consider the coarsened data-structure $(W,A,T, W_T, \Delta, \Delta J)$. By Rose, van der Laan (2011)\nocite{RosevanderLaan2011Missing} (See also van der Laan, Robins (2003)), the inverse probability of coarsening/missingness (IPC)-weighted influence function
 $$D_{IPCW,P}(o):= \frac{\delta}{\Pi_P(w_T,t,a,w)}  \Lambda_P \widetilde H_P(a,w,t)[j - \mu_P(a,w,t)]  $$
 is a gradient for the parameter $\beta_{adj}(P)$ under $ \mathcal{M}_{adj}$ (allowing for missingness). 
 
 Since $\beta_{adj}$ only depends on the conditional distribution of $\Delta J \mid W_T, T, A, W$ and $W_T \mid T, A, W$, the efficient influence function of $\beta_{adj}$ under $ \mathcal{M}_{adj}$ is contained in the direct sum of the tangent spaces $T_{\Delta J} \mathcal{M}_{adj}$ and $T_{W_T} \mathcal{M}_{adj}$ of the statistical model $ \mathcal{M}_{adj}$ obtained by taking scores along regular paths (or submodels) that fluctuate these conditional distributions (Bickel et al., 1993, \cite{vanderlaanunified}). These tangent spaces are contained in the tangent spaces of a nonparametric statistical model $\mathcal{M}_{np}$ given by
 \begin{align*}
T_{\Delta J}\mathcal{M}_{np} &:= \left\{o \mapsto h(w,a,t,w_T,\delta, \delta j) \in L^2(P): E[h(W,A,T, W_T,\Delta, \Delta J)|W,A,T, W_T, \Delta]=0 \right\},\\
T_{W_T}\mathcal{M}_{np} &:= \left\{o \mapsto h(w,a, t, w_T) \in L^2(P): E[h(W,A,T, W_T)|X,A]=0 \right\}.
 \end{align*}

The orthogonal projection of the IPCW gradient $D_{IPCW}$ onto these tangent spaces is also a gradient for $\beta_{adj}$; This follows since the projection only removes components of $D_{IPCW,P}$ that are orthogonal to the direct sum of the tangent spaces $T_{\Delta J}\mathcal{M}_{np}$ and $T_{W_T}\mathcal{M}_{np}$, which is contained in the orthogonal compliment of the direct sum of the tangent spaces $T_{\Delta J} \mathcal{M}_{adj}$ and $T_{W_T} \mathcal{M}_{adj}$. Now, the projections onto the two nonparametric tangent space \citep{vanderlaanunified} are known to be
\begin{align*}
    h \mapsto \Pi_{\Delta J}(h) &= h(W,A,T, W_T,J) - E[h(W,A,T, W_T,J)|W,A,T, W_T];\\
    h \mapsto \Pi_{W_T}(h) &= h(W,A,T, W_T) - E[h(W,A,T, W_T)|W,A].
\end{align*}
Computing the projections gives
\begin{align*}
    \Pi_{\Delta J} D_{IPCW}(O) + \Pi_{W_T} D_{IPCW}(O)  & = \frac{\Delta}{\Pi(W_T,A,T,W)}  \Lambda_{adj,P} \widetilde H_{adj, P} (A,T,W)\left[J - \overline{\mu}(W_T, A,T,W) \right]\\
 & \quad + \Lambda_{adj,P} \widetilde H_{adj, P} (A,T,W)\left[\overline{\mu}(W_T, A,T,W) - \mu_{adj,P}(A,T,W) \right].
\end{align*}
It follows that a gradient in the semiparametric statistical model $ \mathcal{M}_{adj}$ is, viewed as a function of the random variable $O$, given by
\begin{align*}
    D_{adj,P}  & :=   \frac{\Delta}{\Pi(W_T,A,T,W)}  \Lambda_{adj,P} \widetilde  H_{adj, P} (A,T,W)\left[J - \overline{\mu}(W_T, A,T,W) \right]\\
& \quad + \Lambda_{adj,P} \widetilde  H_{adj, P} (A,T,W)\left[\overline{\mu}(W_T, A,T,W) - \mu_{\beta_{adj},h_{adj,P}}(A,T,W) \right].
\end{align*}
This completes the proof.

 \end{proof}

  \section{Proof of limiting distributions of TMLEs}

Let $\widehat{P}_{n}^*$ be the targeted estimator of $P_0$ obtained using the method of Section \ref{section::spTMLE}. Similarly, let $\widehat{P}_{n,adj}^*$ be the targeted distribution obtained using the method of Section \ref{section::spTMLE}. Let $\Lambda_{P_n^*}$ and $\Lambda_{adj,P_n^*}$ be estimators of the scaling matrices $\Lambda_{P}$ and $\Lambda_{adj,P}$ that are compatible with $\widehat P_{n}^*$ and $\widehat P_{n,adj}^*$.

\begin{proof}[Proof of Theorem \ref{theorem::limitspTMLE}]
    The following proof is standard and we refer to \cite{vanderLaanRose2011} for details. We have the expansion:
    \begin{align*}
        \beta(P_n^*) - \beta(P)= - P_n D_{P_{n}^*} +  (P_n -  P) D_{P_{n}^*} + \left\{ \beta(P_n^*) - \beta(P)+ P D_{P_{n}^*} \right\}.
    \end{align*}
    The first term is $o_P(n^{-1/2})$ by construction of the TMLE. The third term is $o_P(n^{-1/2})$ by Lemma \ref{lemma::secondOrder1} given after this proof. Thus, we have the expansion:
    \begin{align*}
        \beta(P_n^*) - \beta(P)&=  (P_n -  P) D_{P_{n}^*} + o_P(n^{-1/2})\\
       & =  (P_n -  P) \left\{D_P \right\} +  (P_n -  P) \left\{D_{P_{n}^*} - D_P \right\} + o_P(n^{-1/2}).
    \end{align*}
    Under Condition  \ref{cond::limitspTMLEaBounded}, $D_P$ is uniformly bounded and thus has finite variance. By the CLT, we have $\sqrt{n}(P_n -  P) D_P $ goes to a normally distributed random variable with the desired covariance structure. Thus, by Slutsky's lemma, it suffices to show that 
    $$(P_n -  P) \left\{D_{P_{n}^*} - D_P  \right\} = o_P(n^{-1/2}).$$
    By Condition \ref{cond::limitspTMLEMissinga} and the invariance of Donsker classes under Lipschitz transformation \citep{vanderVaartWellner}, we have that $D_{P_{n}^*} - D_P$ falls in a Donsker function class. Moreover, Condition \ref{cond::limitspTMLEMissingb} implies that $\norm{D_{P_{n}^*} - D_{P_{n}^*}} = o_P(1)$. By asymptotic equicontinuity of empirical processes over Donsker classes, we have that $(P_n -  P) \left\{D_{P_{n}^*} - D_P \right\} = o_P(n^{-1/2})$ as desired. The result now follows.
\end{proof}

\begin{proof}[Proof of Theorem \ref{theorem::limitspTMLEMissing}]
This proof follows by an argument identical to that of Theorem \ref{theorem::limitspTMLEMissing} after substituting Lemma  \ref{lemma::secondOrder1} with Lemma  \ref{lemma::secondOrder2} in the argument.
\end{proof}

\begin{lemma}
Under the conditions of Theorem \ref{theorem::limitspTMLE}, we have the following exact second-order remainder satisfies
 $$\beta(P_n^*) - \beta(P)+ P D_{P_n^*} = o_P(n^{-1/2}).$$
 \label{lemma::secondOrder1}
\end{lemma}

\begin{lemma}
Under the conditions of Theorem \ref{theorem::limitspTMLEMissing}, we have the following exact second-order remainder satisfies
 $$\beta_{adj}(P_{n,adj}^*) - \beta_{adj} + P D_{\widehat P_{n, adj}^*} = o_P(n^{-1/2}).$$
 \label{lemma::secondOrder2}
\end{lemma}

\begin{proof}[Proof of Lemma  \ref{lemma::secondOrder1}]

 For ease of notation, let $\beta_{n}:= \beta(P_n^*) $, $h_n := h_{P_n^*}$, $ \mu_n :=  \mu_{ P_n^*}$, and $\sigma_n^2 = \mu_n (1 - \mu_n)$. Let $H_{P}(T,A,W; \mu_n, \widetilde{\pi}_P)$
 
 Let $P_{n,0}^*$ be t
 Furthermore, for $P' \in \mathcal{M}_1$, we define $\sigma^2_{\mu} := \mu \{1 - \mu\}$ and
 $$ H(a,w, t; \mu, \widetilde{\pi} ) :=   \left(a  -\frac{ \widetilde{\pi}(w,t) \} 
}{  \widetilde{\pi}(w,t) \sigma^2_{\mu}(1,w,t)  +  (1-\widetilde{\pi}(w,t)) \sigma^2_{\mu}(1,w,t) }\right),$$
 where we note that $ H(a,w, t; \mu_{P'}, \widetilde{\pi}_{P'} ) = H_{P'}(a,w, t)$. Let $\mathcal{I}_{\Delta}: o \mapsto \delta$ be the coordinate projection identity map for the missingness indicator.

 By the law of iterative expectations, the second order remainder can be written as:
  \begin{align*}
      \beta_{n} - \beta(P) & + P D_{P_n^*} \\
  &  =  \beta_{n,} - \beta(P)  + P\left[\left\{ \mathcal{I}_{\Delta}\Lambda_{P_n^*} \underline{f}  \right\}H_{\widehat P_{n}^*}   \left\{\mu_P   -   {\mu}_{n} \right\}\right] . 
 \end{align*}
Now, recall that
 $$\left[ \mu_P(T,A,W)  - \mu_{n}(T,A,W)\right] = \text{expit}\left\{h_{P}(W,T) + A\underline{f}(W,T)^\top \beta(P)\right\} - \text{expit}\left\{h_n(W,T) + A\underline{f}(W,T)^\top \beta_n \right\}.$$
 Since $h_n$, $\underline{f}$ and $\beta_n$ are uniformly bounded by Conditions  \ref{cond::limitFeature} and  \ref{cond::limitspTMLEaBounded}, we have the second-order Taylor expansion:
 \begin{align*}
     \left[ \mu_P(T,A,W)  - \mu_{n}(T,A,W)\right] &= \sigma_n^2(A,T,W)\left[h_{P}(W,T) - h_n(W,T) + A\underline{f}(W,T)^\top (\beta(P)- \beta_n)\right] \\
     & \quad + O(|h_{P}(W,T) - h_n(W,T)|^2 + \norm{\beta(P)- \beta_n}^2).
 \end{align*}
 
In the following we will abuse notation and implicitly condition on the training data used to obtain the nuisance estimators. On the pages 621-622 and 626-629 of the working paper \cite{OddsRatioreadingsTMLE}, the following orthogonality condition was derived: 
 $$E_P[\sigma^2_n(A,T,W)\cdot H(T,A,W; \mu_n, \widetilde{\pi}_P) \cdot h(W,T) \mid \Delta =1, T, W] = 0 \text{ for all } h \in L^2(P_{T,W}).$$
  It follows that 
  \begin{align*}
      &\hspace{-1cm} E_P[\sigma^2_n(A,T,W)\cdot H(T,A,W; \mu_n, \widetilde{\pi}_n) \cdot h(W,T) \mid \Delta =1, T, W]  \\
      &= E_P[\sigma^2_n(A,T,W)\cdot \left\{H(T,A,W; \mu_n, \widetilde{\pi}_n)  - H(T,A,W; \mu_n, \widetilde{\pi}_P) \right\}\cdot h(W,T) \mid \Delta =1, T, W]
  \end{align*} 
  which is $O(|\widetilde{\pi}_n(W,T) - \widetilde{\pi}_P(W,T)|)$ under Condition \ref{cond::limitspTMLEaBounded}. Taking $h = h_{P} - h_n$ in the orthogonality condition and plugging in the Taylor expansion, it follows that 
  \begin{align*}
   & \hspace{-1cm}  E\left\{ \Lambda_{P_n^*} \underline{f}(W,T) \right\}H_{\widehat P_{n}^*}(T,A,W)  \left[\mu_P(T,A,W)  -   {\mu}_{n}(T,A,W) \mid \Delta =1, T, W\right] \\
     &= \sigma_n^2(A,T,W)\left[ A\underline{f}(W,T)^\top (\beta(P)- \beta_n)\right] \\
     & \quad + O(|h_{P}(W,T) - h_n(W,T)||\widetilde{\pi}_n(W,T) - \widetilde{\pi}_P(W,T)|) +\\
     & \quad + O(|h_{P}(W,T) - h_n(W,T)|^2 + \norm{\beta(P)- \beta_n}^2).
 \end{align*}
 It follows that second-order remainder is equal to
 \begin{align*}
   &=  \beta_n - \beta(P) + (\beta(P)- \beta_n) E_P \left[ \Delta A \Lambda_{P_n^*} \underline{f}(W,T)\underline{f}(W,T)^\top H(T, A, W; \mu_n, \widetilde \pi_n)  \sigma_n^2(A,T,W)\right]\\
   & \quad + O_P\left(\norm{\beta(P)- \beta_n}^2 + \norm{\widetilde{\pi}_n - \widetilde{\pi}_P}^2 + \norm{h_{P} - h_n}^2\right).
 \end{align*}
Performing another Taylor expansion, we find the remainder also equals
 \begin{align*}
    &= \beta_n - \beta(P)+ (\beta(P)- \beta_n) E_P \left[ \Delta A \Lambda_{P} \underline{f}(W,T)\underline{f}(W,T)^\top H_{P}(T, A, W)  \sigma_P^2(A,T,W)\right]\\
     & \quad + O_P\left(\norm{\beta(P)- \beta_n}^2 + \norm{\widetilde{\pi}_n - \widetilde{\pi}_P}^2 + \norm{h_{P} - h_n}^2\right).
 \end{align*}
Computing the expectation, we can show that
$$E_P \left[ \Delta A \Lambda_{P} \underline{f}(W,T)\underline{f}(W,T)^\top H_{P}(T, A, W)  \sigma_P^2(A,T,W)\right] = 1
$$
which causes the leading term of the previous display to vanish. Thus, the second-order remainder
 \begin{align*}
    = O_P\left(\norm{\beta(P)- \beta_n}^2 + \norm{\widetilde{\pi}_n - \widetilde{\pi}_P}^2 + \norm{h_{P} - h_n}^2\right).
 \end{align*}
 We claim that 
 $$\norm{h_{P} - h_n}^2 = O_p(\norm{\mu_n(a=0, \cdot) - \mu_P(a=0, \cdot)}^2)$$ and 
 $$\norm{\beta_{n} - \beta(P)}^2 = O_p\left(\max_{a \in \{0,1\}}\norm{\mu_n(a, \cdot) - \mu_P(a, \cdot)}^2\right).$$ 
 If these claims are hold, Condition \ref{cond::limitspTMLEc} implies that $\norm{h_{P} - h_n}^2$ and $\norm{\beta_{n} - \beta(P)}^2 $ are $o_P(n^{-1/2})$. This would imply the same for the second-order remainder, thereby completing the proof. The first claim holds since 
$$\left|h_{P}(w) - h_n(w)  \right| = \left|\logit\mu_n(a=0,w) - \logit \mu_P(a=0,w)  \right| \lessapprox \left| \mu_n(a=0,w) - \mu_P(a=0,w)\right|,$$
where the final inequality applies Condition  \ref{cond::limitspTMLEaBounded} and Lipschitz continuity of the logit map on $[\delta, 1-\delta] \subset (0,1)$.
For the second claim, we have
$$ \underline{f}(\cdot)^T(\beta_n - \beta(P))  =  \left\{ \logit\mu_n(a=1,\cdot) - \logit\mu_P(a=1,\cdot) \right\} - \left\{ \logit\mu_n(a=0,\cdot) - \logit\mu_P(a=0,\cdot) \right\}.$$
Thus, using Lipschitz continuity of the logit function, squaring, and taking expectations, we find
$$ E\left|\underline{f}(\cdot)^T(\beta_n - \beta(P)) \right|^2  \lessapprox  \max_{a \in \{0,1\}} \norm{\mu_n(a, \cdot) - \mu_P(a, \cdot)}^2.$$
 Next, observe that $|\underline{f}(\cdot)^T(\beta_n - \beta(P))|^2  =    |(\beta_n - \beta(P))^T\underline{f}(\cdot)\underline{f}(\cdot)^T(\beta_n - \beta(P))|.$ Thus,
 $$ E\left|\underline{f}(\cdot)^T(\beta_n - \beta(P)) \right|^2 = |(\beta_n - \beta(P))^TE\left[\underline{f}(\cdot)\underline{f}(\cdot)^T\right](\beta_n - \beta(P))|.  $$
 Since $E\left[\underline{f}(\cdot)\underline{f}(\cdot)^T\right]$ is necessarily positive definite and also invertible by Condition   \ref{cond::limitFeature}, we have 
$$ |(\beta_n - \beta(P))^TE\left[\underline{f}(\cdot)\underline{f}(\cdot)^T\right](\beta_n - \beta(P))| \geq \gamma \norm{\beta_n - \beta(P)}_2^2, $$
where $ \gamma$ is smallest eigenvalue of  $E\left[\underline{f}(\cdot)\underline{f}(\cdot)^T\right]$. Putting it all together, we conclude that
$$ \norm{\beta_n - \beta(P)}_2^2  = O_P\left( \max_{a \in \{0,1\}} \norm{\mu_n(a, \cdot) - \mu_P(a, \cdot)}^2\right).$$

\end{proof}

\begin{proof}[Proof of Lemma  \ref{lemma::secondOrder2}]
For ease of notation, we drop the dependence on $adj$ in some of our notation. For example, we simply write $P_n^*$ instead of $P_{n,adj}^*$. We have
    \begin{align*}
  &   D_{ P_{n}^*, adj}(O)=     \frac{\Delta}{\Pi_n(W_T,A,W)} \left\{\Lambda_{adj,P_n^*} \underline{f}(W,T) \right\}   H_{adj,P_n^*}(T,A,W )  \left[ J  -  \overline{\mu}_n^*(W_T,T,A,W) \right] \\
& +   \left\{ \Lambda_{adj,P_n^*} \underline{f}(W,T) \right\}H_{adj,P_n^*}(T,A,W)  \left[\overline{\mu}_n^*(W_T,T,A,W)  -  {\mu}_{n}^*(T,A,W) \right].
    \end{align*}
We will also implicitly condition on the training data used to obtain the nuisance estimators when writing expectations of the form $E[\cdot]$, so that expectations only average over the random variable $O = (\Delta J, \Delta, W_T, T, A, W)$.

 Taking the expectation over $O$ and applying the law of iterated expectations twice, we find

\begin{align*}
  & E \left[ \frac{\Delta}{\Pi_n(W_T,A,W)} \left\{\Lambda_{adj,P_n^*} \underline{f}(W,T) \right\}   H_{adj,P_n^*}(T,A,W )  \left[ J  -  \overline{\mu}_n^*(W_T,T,A,W) \right] \right] \\
& =   E \left[ \frac{\Delta}{\Pi_n(W_T,A,W)} \left\{\Lambda_{adj,P_n^*} \underline{f}(W,T) \right\}   H_{adj,P_n^*}(T,A,W )  \left[ \overline{\mu}(W_T, T, A ,W)  -  \overline{\mu}_n^*(W_T,T,A,W) \right] \right]\\
& =   E \left[ \frac{\Pi(W_T,A,W)}{\Pi_n(W_T,A,W)} \left\{\Lambda_{adj,P_n^*} \underline{f}(W,T) \right\}   H_{adj,P_n^*}(T,A,W )  \left[ \overline{\mu}(W_T, T, A ,W)  -  \overline{\mu}_n^*(W_T,T,A,W) \right] \right].
\end{align*}
Thus, we have 
\begin{align*}
    &E \left[ \frac{\Delta}{\Pi_n(W_T,A,W)} \left\{\Lambda_{adj,P_n^*} \underline{f}(W,T) \right\}   H_{adj,P_n^*}(T,A,W )  \left[ J  -  \overline{\mu}_n^*(W_T,T,A,W) \right] \right]\\ &=   E \left[ \left[1 - \frac{\Pi(W_T,A,W)}{\Pi_n(W_T,A,W)} \right] \left\{\Lambda_{adj,P_n^*} \underline{f}(W,T) \right\}   H_{adj,P_n^*}(T,A,W )  \left[ \overline{\mu}(W_T, T, A ,W)  -  \overline{\mu}_n^*(W_T,T,A,W) \right] \right] \\
    &\quad +  E \left[  \left\{\Lambda_{adj,P_n^*} \underline{f}(W,T) \right\}   H_{adj,P_n^*}(T,A,W )  \left[ \overline{\mu}_P(W_T, T, A ,W)  -  \overline{\mu}_n^*(W_T,T,A,W) \right] \right].
\end{align*}
By Condition \ref{cond::limitspTMLEMissingpos} and Cauchy-Schwarz, the first term on the RHS is $O_P(\norm{\Pi_n - \Pi}\norm{\overline{\mu}_n^* - \overline{\mu}_P})$. Thus, returning to our first display, we find
 \begin{align*}
    P D_{;\widehat P_{n}^*, adj} &= 
      E \left[  \left\{\Lambda_{adj,P_n^*} \underline{f}(W,T) \right\}   H_{adj,P_n^*}(T,A,W )  \left[ \overline{\mu}_P(W_T, T, A ,W)  -  \overline{\mu}_n^*(W_T,T,A,W) \right] \right]\\
     & \quad \quad  + O_P(\norm{\Pi_n - \Pi}\norm{\overline{\mu}_n^* - \overline{\mu}_P})  \\
     & \quad \quad +  E \left[\left\{ \Lambda_{adj,P_n^*} \underline{f}(W,T) \right\}H_{adj,P_n^*}(T,A,W)  \left[\overline{\mu}_n^*(W_T,T,A,W)  -  {\mu}_{n}^*(T,A,W) \right]\right]
    \\
   &=      E\left\{ \Lambda_{adj,P_n^*} \underline{f}(W,T) \right\}H_{adj,P_n^*}(T,A,W)  \left[\overline{\mu}_P(W_T,T,A,W)  -  \mu_n^*(T,A,W) \right]\\
   & \quad + O_P(\norm{\Pi_n - \Pi}\norm{\overline{\mu}_n^* - \overline{\mu}_P}) \\
    &=        E\left\{ \Lambda_{adj,P_n^*} \underline{f}(W,T) \right\}H_{adj,P_n^*}(T,A,W)  \left[\mu_P(T,A,W)  -   {\mu}_n^*(T,A,W) \right]\\
     & \quad + O_P(\norm{\Pi_n - \Pi}\norm{\overline{\mu}_n^* - \overline{\mu}_P}) ,
    \end{align*}
where the final equality follows from the law of iterated conditional expectations.  
Since Condition \ref{cond::limitspTMLEMissingb} implies $O_P(\norm{\Pi_n - \Pi}\norm{\overline{\mu}_n^* - \overline{\mu}_P}) = o_P(n^{-1/2})$, it suffices to show that
 \begin{align*}
    \hspace{-0.5cm} & \beta_{adj}(P_n^*) - \beta_{adj}(P) + E\left\{ \Lambda_{adj,P_n^*} \underline{f}(W,T) \right\}H_{adj,P_n^*}(T,A,W)  \left[\mu_P(T,A,W) -   {\mu}_n^*(T,A,W) \right] \\
      &= \beta_{adj}(P_n^*) - \beta_{adj}(P) + E\left\{ \Lambda_{adj,P_n^*} \underline{f}(W,T) \right\}H_{adj,P_n^*}(T,A,W)  \left[J  -   {\mu}_n^*(T,A,W) \right] \\
      &= o_P(n^{-1/2}),
    \end{align*}
This follows noting that the RHS is of the same form as the second-order remainder of Lemma \ref{lemma::secondOrder1}. A proof identical to that of Lemma  \ref{lemma::secondOrder1} establishes the desired result.

\end{proof}

\end{document}